%% file: main.tex
\documentclass[twoside]{article}

\usepackage[accepted]{aistats2023}
\usepackage[round]{natbib}

\usepackage[utf8]{inputenc} %
\usepackage[T1]{fontenc}    %
\usepackage{hyperref}       %
\usepackage{url}            %
\usepackage{booktabs}       %
\usepackage{amsfonts}       %
\usepackage{nicefrac}       %
\usepackage{microtype}      %
\usepackage{xcolor}         %
\usepackage{adjustbox}
\usepackage{multirow}
\RequirePackage{algorithm}
\RequirePackage{algorithmic}
\usepackage{wrapfig}
\usepackage{parskip}

\usepackage{graphicx}
\usepackage{subcaption}
\usepackage{breqn}
\usepackage{tabularx}

\usepackage{amsmath}
\usepackage{amssymb}
\usepackage{mathtools}
\usepackage{amsthm}
\usepackage{thmtools} 
\usepackage[font=small,labelfont=bf]{caption}
\usepackage{makecell}
\usepackage{wrapfig}

\usepackage[textsize=tiny]{todonotes}

\usepackage[capitalize]{cleveref}
\usepackage{cancel}

\theoremstyle{plain}

\newtheorem{lemma}{Lemma}[section]

\crefname{corollary}{Corollary}{Corollaries}
\theoremstyle{definition}
\newtheorem{definition}{Definition}[section]
\crefname{definition}{definition}{definitions}
\newtheorem{assumption}{Assumption}[section]
\crefname{assumption}{Assumption}{Assumptions}

\crefname{conjecture}{conjecture}{conjectures}
\newtheorem{example}{Example}[section]
\crefname{example}{Example}{Examples}
\theoremstyle{remark}
\newtheorem{remark}{Remark}[section]
\crefname{remark}{Remark}{Remarks}

\usepackage{shortcuts}
\newcommand{\mko}[1]{{}}

\newcommand{\ilkrm}[1]{}
\newcommand{\zmh}[1]{{\color{blue}Z\@: #1}}
\newcommand{\zmhrm}[1]{}

\input{paperSpecificNotation}
\usepackage{csquotes}
\usepackage[inline]{enumitem} 

\newcommand\Ccancel[2][black]{\renewcommand\CancelColor{\color{#1}}\cancel{#2}}

\newenvironment{centermath}
 {\begin{center}$\displaystyle}
 {$\end{center}}

\begin{document}

\runningauthor{Zeshan Hussain*, Ming-Chieh Shih*, Michael Oberst, Ilker Demirel, David Sontag}

\twocolumn[

\aistatstitle{Falsification of Internal and External Validity in Observational Studies via Conditional Moment Restrictions}

\aistatsauthor{ Zeshan Hussain* \And Ming-Chieh Shih* \And  Michael Oberst}
\aistatsaddress{ MIT \And  National Dong Hwa University \And MIT} 
\aistatsauthor{ Ilker Demirel \And David Sontag}
\aistatsaddress{MIT \And MIT}
]

\begin{abstract}
\textit{Randomized Controlled Trials (RCT)s} are relied upon to assess new treatments, but suffer from limited power to guide personalized treatment decisions. On the other hand, observational (i.e., non-experimental) studies have large and diverse populations, but are prone to various biases (e.g. residual confounding). To safely leverage the strengths of observational studies, we focus on the problem of \textit{falsification}, whereby RCTs are used to validate causal effect estimates learned from observational data. In particular, we show that, given data from both an RCT and an observational study, assumptions on internal and external validity have an observable, testable implication in the form of a set of \textit{Conditional Moment Restrictions (CMRs)}. Further, we show that expressing these CMRs with respect to the causal effect, or \enquote{causal contrast}, as opposed to individual counterfactual means, provides a more reliable falsification test. In addition to giving guarantees on the asymptotic properties of our test, we demonstrate superior power and type I error of our approach on semi-synthetic and real world datasets. Our approach is interpretable, allowing a practitioner to visualize which subgroups in the population lead to falsification of an observational study.
\end{abstract}

\section{INTRODUCTION}

Observational studies, prevalent in healthcare, economics, and other fields, are an important source of real-world data used to derive granular treatment effect estimates \citep{dagan2021bnt162b2,hernan2008observational,imbens2009recent,xie2020heterogeneous}. Indeed, there has been a rich literature in building estimators of heterogeneous treatment effects from observational data, particularly using modern machine learning methods \citep{wager2018estimation,kunzel2019metalearners,semenova2021debiased}. However, observational studies may lack \textit{internal validity} in that estimates of causal effects (in the observational population) may be biased or inconsistent, e.g., due to unobserved differences between the treatment and control groups, such as in the setting of unobserved confounding. %
On the other hand, observational studies are representative of more diverse populations, leading to more plausible \textit{external validity}, i.e., ability to generalize estimates across wider populations.
In contrast, \textit{Randomized Controlled Trials (RCTs)} have strong internal validity, assuming sound design (e.g. a prospective trial, a-priori definition of hypotheses to be tested) and appropriate randomization. 
However, RCTs often have restrictive inclusion criteria, which can call their external validity into question \citep{degtiar2021review}, and are of limited size, limiting their ability to detect differences in treatment effect for specific sub-populations, or detect differences in adverse event rates (if e.g., the adverse event is rare) 
\citep{tsang2009inadequate,ali2018sample,phillips2019analysis}.
Intuitively, we would like to leverage observational data for estimating treatment effects that cannot be reliably estimated using RCTs, whether due to a lack of statistical power or a lack of patient diversity in the RCT\@.  At the same time, we would like to take advantage of the strong internal validity of the RCT to increase confidence in our observational estimates.  %

With these considerations in mind, we study the problem of using limited RCT data to \enquote{falsify} assumptions of internal and external validity for observational studies. Our method can be applied even when the RCT data does not cover the entire observational population, and hence cannot be used on its own to estimate causal effects.
Assuming that the RCT has internal validity, we show that assumptions of internal/external validity of the observational study have a testable implication in the form of a set of conditional moment restrictions (CMRs).
We propose a \textit{falsification algorithm} that tests whether or not these CMRs hold, thereby providing an opportunity to reject these assumptions when they fail to hold. 
This allows us to take advantage of approaches developed in the econometrics literature for testing CMRs, and we use a Maximum Moment Restriction (MMR)-based test \citep{muandet2020kernel} for this purpose.

Compared to prior work, the benefits of our approach are two-fold. First, we implicitly check across all subpopulations of covariates $X$ for disagreement between the conditional average treatment effect (CATE) functions as estimated in the RCT versus in the observational study. Second, as an additional benefit, our approach provides an explanation for rejection, in the form of a \enquote{witness function}, which describes subpopulations where these estimates diverge.

Importantly, in constructing the test for our problem, we use the insight that not all differences between observational and experimental distributions matter. For instance, there may be differences in unobserved baseline risk factors, which cause estimates of individual potential outcomes to differ, but do not impact the causal effect (the difference between control and treated outcomes). A naive MMR-based test would asymptotically reject in this scenario, but we demonstrate that, by careful construction of the MMR test statistic, we can avoid this failure case.

Our method can be compared to prior approaches to falsification of observational studies. One such approach is to check for a statistically significant difference between estimates of the average treatment effect (ATE) from the RCT and from the observational study \citep{franklin2021emulating,dagan2021bnt162b2,baden2020efficacy}. Unfortunately, this approach can lead to false negatives, e.g., if the ATE from the observational study replicates the RCT ATE, despite biases on finer-grained subpopulations. Furthermore, even if an observational study is correctly \enquote{rejected}, the approach does not provide an \textbf{explanation} for why the observational study was rejected, which is an important practical consideration for both statisticians and policymakers. Another approach is to compare subgroup-level effects instead of the ATE. \citet{hussain2022falsification} adopt this approach in the context of testing (multiple) observational estimates against RCT estimates, essentially testing for differences in group-wise treatment effects. However, this approach requires correction for multiple hypothesis testing across subgroups, and a-priori specification of these subgroups, which can limit its ability to uncover areas of disagreement. %

\textbf{Contributions}: We have the following desiderata for our falsification algorithm:
\begin{enumerate*}[label= (\roman*)]
  \item rejecting observational studies when their underlying causal assumptions fail \textit{(high power)}, 
  \item accepting in cases where these casual assumptions hold \textit{(controlled type I error)}, and
  \item providing an explanation of why an observational study is rejected. 
\end{enumerate*}
With these desiderata in mind, our main contributions are as follows:
\begin{enumerate*}[label= (\roman*)]
  \item First, we show how to convert causal assumptions on internal and external validity into a set of CMRs, violations of which can be detected using observational and RCT data using existing techniques with theoretical guarantees.
  \item Second, we demonstrate that our construction of these CMRs avoids a potential failure mode: rejecting observational studies due to differences in unobserved covariates that influence baseline outcomes, but not treatment effects.
  \item Third, on semi-synthetic and real-world datasets, we show favorable performance of our method with respect to power and type I error, and showcase its ability to produce informative \textit{explanations} of rejections. 
\end{enumerate*}

\section{SETUP AND MOTIVATING EXAMPLES}

\subsection{Notation \& Assumptions}

Let $Y\in \cY$ be the outcome of interest, and $A\in \{0,1\}$ denote a binary treatment variable. We let $Y_a$ denote the potential outcome under treatment $A = a$, and we use $X \in \cX$ to denote the full set of covariates. 
Note that, in our development, we operate in the setting where there is a single observational study and RCT.  
We use an indicator variable, $S = \{0,1\}$, where $S=1$ denotes data from the observational study and $S=0$ from the RCT. 

To further characterize the observational study and RCT, we let $\cI_0$ and $\cI_1$ be the observed indices for the RCT and observational study, respectively. Furthermore, we let $\cI = \cI_0 \cup \cI_1$ be the total set of observed indices. We use $|\cI|$ to denote the cardinality of a set, and let $|\cI_0| = n_0$, $|\cI_1| = n_1$, and $|\cI| = n$. Finally, $\E[.]$ and $\P[.]$ are expectations and probabilities taken with respect to the joint distribution $\P(Y, A, X, S)$ of the observational study and RCT.

Our goal is to discover violations of causal assumptions that underlie the validity of conditional average treatment effect (CATE) estimates derived from the observational study. To that end, we first state these assumptions formally. %
\begin{restatable}[\textit{Internal Validity of Observational Data}]{assumption}{Internal Validity of Observational Data}\label{asmp:support}
We assume the following in the observational study: 
\begin{itemize}[leftmargin=*]
    \item \textit{Ignorability} — $Y_a \indep A \mid  X, (S = 1)$,~$\forall a \in \{0,1\}$.
    \item \textit{Consistency} — $A = a, S =1 \!\! \implies \!\! Y_a = Y$,~$\forall a \in \{0,1\}$.
    \item \textit{Positivity of Treatment} — $\P(X = x, S = 1) > 0 \implies \P(A = a | X = x, S = 1) > 0$,~$\forall a \in \{0,1\}$ and $\forall x \in {\cal X}$.
\end{itemize}
\end{restatable}

\Cref{asmp:support} gives a standard set of assumptions under which the CATE conditioned on $X$, $\E[Y_1-Y_0|X = x, S=1]$ can be identified.  However, this assumption is not testable in isolation.  In order to compare observational estimates with those of the RCT to discover flaws, we will first need to assume that the RCT itself provides valid estimates.

\begin{restatable}[\textit{Internal Validity of RCT}]{assumption}{Internal Validity of RCT}\label{asmp:rct-validity}
We assume the following in the RCT: 
\begin{itemize}[leftmargin=*]
    \item \textit{Ignorability} — $Y_a \indep A \mid  X, (S = 0)$,~$\forall a \in \{0,1\}$.
    \item \textit{Consistency} — $A = a, S =0 \!\! \implies \!\! Y_a = Y$,~$\forall a \in \{0,1\}$.
    \item \textit{Fixed probability of assignment} — $P(A = 1 | X = x, S = 0) = p$, for some $p \in (0,1)$, $\forall x \in {\cal X}$.
\end{itemize}
\end{restatable}

\Cref{asmp:rct-validity} is a generally defensible (and standard) set of assumptions on the validity of the RCT. However, even if both~\cref{asmp:rct-validity,asmp:support} hold, the corresponding CATE functions are not necessarily comparable.
For instance, there may be unmeasured effect modifiers that have different distributions between the RCT and observational study. %
Under the following additional assumption, the CATE in the RCT (i.e., $\E[Y_1 - Y_0 \mid X = x, S = 0]$) can be identified using observational data.

\begin{restatable}[\textit{External Validity: Observational Study to RCT Transportability of CATE}]{assumption}{Observational Study to RCT Transportability of CATE}\label{asmp:transport}
We assume the following: 
\begin{itemize}[leftmargin=*]
    \item \textit{Mean Exchangeability of Contrast} — $\E [Y_1 - Y_0 | X = x] = \E [Y_1 - Y_0 | X = x, S = s]$, $\forall x \in {\cal X}$ and $\forall s \in \{0,1\}$.
    \item \textit{Positivity of Selection} — $\P(X = x | S = 0) > 0 \implies \P(X = x | S = 1) > 0 $, $\forall x \in {\cal X}$. 
\end{itemize}
\end{restatable}
The first part of this assumption is sometimes referred to as \enquote{generalizability in effect measure} \citep{dahabreh2019} or \enquote{conditional exchangeability in measure} \citep{dahabreh2020extending}.  This assumption is weaker than (and implied by) transportability of counterfactual means (e.g., $\E[Y_a \mid X, S] = \E[Y_a \mid X])$\footnote{Equality relations including random variables are to be understood as ``almost sure'' (a.s.) relations throughout the manuscript.}.  It is simple to show that this assumption (along with our other assumptions) is sufficient to identify the CATE in the RCT population using the observational distribution alone.

\begin{restatable}[]{proposition}{ThmTransport}\label{thm:transport}
Under~\cref{asmp:support,asmp:transport}, the CATE of the RCT given $X$, $\E[Y_1-Y_0 | X, S = 0]$, is identifiable in the observational data by
\begin{equation}\label{eq:estimand}
\E[Y \mid X, A=1, S= 1] - \E[Y \mid X, A=0, S= 1]
\end{equation}
\end{restatable}
\Cref{thm:transport} follows from substantially the same arguments used by~\citet{dahabreh2019} for identification of average treatment effects under similar assumptions, but we include a short proof in~\cref{sec:app-proofs} for completeness, alongside all other proofs for this paper.

\begin{remark}
As we demonstrate later on, our statistical test does not distinguish between violations of \cref{asmp:support} or \cref{asmp:transport}.  However, violation of either assumption is a meaningful finding when considering the credibility of causal effects learned from observational data.  For instance, even if the observational study is free of unmeasured confounding (i.e., \cref{asmp:support} holds), there may exist unmeasured effect modifiers whose distributions differ substantially across populations, leading to a violation of~\cref{asmp:transport}.  In other words, if the true CATE function varies substantially for individuals with the same covariates $X$ across the observational and RCT populations, then it may not reliably generalize to future patients.
\end{remark}

\subsection{Motivation: Testing for Differences in Causal Contrasts, rather than Counterfactual Means}
\label{sec:id-results}

When it is possible to identify a causal effect from an observational study, we would prefer to avoid rejecting that study unnecessarily. This motivates~\cref{asmp:transport}, which holds even in the scenarios where counterfactual means in the RCT (e.g., the expected outcome under treatment $\E[Y_1 \mid X, S = 0]$) are \textbf{not} identifiable from observational data, but where the causal contrast is identified.

This assumption is central to our testing methodology, as we test for a null hypothesis that is satisfied under~\cref{asmp:transport}, even when counterfactual means are not transportable.  We build intuition for this assumption in two ways. First, we give a structural causal model that formalizes a sufficient condition for this assumption to hold. Second, we give concrete examples of where this assumption appears to (approximately) hold in practice.

\begin{example}\label{example:abs-non-iden}
Let $U$ be a set of variables that are unobserved in both the RCT and observational data. Suppose $Y$ is generated according to the following structural equation, with binary treatment $A$ and observed covariates $X$
\begin{equation}
    Y = g(X, U) + \tau(X) \cdot A + \epsilon_0,
\end{equation}
where $\epsilon_0$ is an independent mean-zero random variable ($\E[\epsilon_0] = 0$ and $\epsilon_0 \indep X, U, A, S$) and where $P(U \mid X, S = 1) \neq P(U \mid X, S = 0)$. 
\end{example}
In~\cref{example:abs-non-iden}, $Y_0 = g(X, U) + \epsilon_0$ is influenced by both $X$ and a set of unobserved baseline characteristics $U$.  As a result, the conditional counterfactual mean $\E[Y_0 \mid X, S] = \E[g(X, U) \mid X, S]$ will generally differ across studies, due to the fact that the distribution of $U$ varies across studies. The conditional average treatment effect, on the other hand, is independent of $U$ and $S$, as $\E[Y_1 - Y_0 \mid X] = \tau(X)$. This quantity is purely a function of $X$, satisfying our assumption that the CATE does not depend on $S$.\footnote{The constant treatment effect for individuals with the same $X$ is not necessary, and merely helps simplify notation.  One could make a similar observation with $\tau(X, \epsilon_{\tau})$ for an additional noise variable $\epsilon_{\tau}$ that is independent of $U, S$.}

This scenario is plausible in real-world settings, where the treatment effect is a function of a subset of variables that influence the outcome $Y$. For a real-world example, consider the SPRINT Trial \citep{sprint2015randomized}, which studies the impact of intensive blood pressure control ($A$) on a composite outcome ($Y$) that includes heart failure and death. Here, previous chronic kidney disease (CKD) is a variable, like $U$, that has a substantial impact on the outcome $Y_0$ under no treatment (as reported in Figure 4 of~\citet{sprint2015randomized}), but does not have a (statistically) significant influence on the treatment effect itself (i.e., $\tau(X)$ in the example above). We discuss this example and other examples of real-world motivation in more detail in~\cref{sec:app-contrast}.

\section{MMR-based FALSIFICATION TESTS}

Next, we observe that~\cref{asmp:rct-validity,asmp:support,asmp:transport} have observable implications on the joint RCT and observational data in the form of a (set of) conditional moment restrictions. As a result, if these restrictions fail to hold, then this implies a violation of the underlying causal assumptions. This suggests a hypothesis-testing approach for detecting violations, which we develop in this section. Notably, the resulting hypothesis test looks for differences between the CATE functions estimated from the RCT and observational studies, but does not test for equality of conditional potential outcomes themselves.  This is motivated from our prior discussion, that conditional means of potential outcomes (e.g., $\E[Y_0 \mid X]$) could differ between the RCT and observational data, even when the CATE function itself is identified.

\subsection{CATE Estimation} \label{sec:cate-est}

The crux of our methodology is to use the CATE estimate from the RCT as a proxy for the true CATE function to falsify or validate an observational estimate. To that end, we first construct an unbiased CATE estimator from RCT data. Since the probability of assignment to each treatment is known by design in RCTs, we can use an IPW-style estimator for the CATE. Similar estimators can be found in standard causal inference textbooks (e.g. Ch. 2 of \cite{Hernan2021-yd}). A ``doubly robust'' variant can be used, but if the outcome model is misspecified, this may result in higher variance and a loss of power in our test. Thus, we first define the following
\enquote{signal} function,
\begin{align} \label{eq:rct_signal}
\psi_0 &= \frac{\1{S = 0}}{P(S = 0 \mid X)} Y \nonumber \\
& \times \left(\frac{\1{A = 1}}{P(A = 1 \mid S = 0)} - \frac{\1{A = 0}}{P(A = 0 \mid S = 0)}\right),
\end{align}
and then observe that the conditional expectation of this signal $\E[\psi_0 \mid X]$ is equal to the CATE \textit{in the RCT population}, using data from the RCT alone.\footnote{Note the use of the indicator $\1{S = 0}$, such that $\psi_0$ only depends on data from the RCT itself, even though we take the conditional expectation over the combined sample.}

\begin{restatable}[\textit{CATE signal from the RCT}]{proposition}{CATESignalRCT}\label{prop:rct-signal}
    Under \cref{asmp:rct-validity}, the instance-wise CATE signal $\psi_0$ in~\cref{eq:rct_signal}, which uses the outcome information from the RCT, is unbiased, i.e., $\E[\psi_0 | X] = \E[Y_1 - Y_0 | X, S = 0]$.
\end{restatable}

Next, we wish to develop a distinct estimate of the CATE in the RCT population, but one which makes use of the observational data. The first step in building such an estimator is to identify, under our causal assumptions, the corresponding \textit{statistical estimand}, i.e.~\cref{eq:estimand} in \cref{thm:transport}. Our goal is to check the \textit{validity} of this estimand, which amounts to challenging~\cref{asmp:support,asmp:transport}. Drawing from existing literature \citep{dahabreh2019,dahabreh2020extending,degtiar2021review}, we employ the following doubly robust signal, which combines response surface modeling and inverse probability weighting (IPW):
\begin{align}
&\psi_1 = \frac{1}{P(S=0 \mid X)} \bigg[\1{S = 0} \underbrace{\big( \mu_1(X) - \mu_0(X) \big)}_{\text{Response Surface Signal}} \nonumber \\
 &+ \1{S=1} \frac{P(S=0 \mid X)}{P(S=1 \mid X)} \bigg( \underbrace{\frac{\1{A=1} (Y- \mu_1(X))}{P(A=1 \mid S=1, X)}}_{\text{IPW Signal}} \nonumber \\
 & - \underbrace{\frac{\1{A=0}(Y-\mu_0(X))}{P(A=0 \mid S=1, X)}}_{\text{IPW Signal}} \bigg) \bigg], \label{eq:obs_signal}
\end{align}
where $\mu_a(X) \coloneqq \E[Y \mid A = a, X, S = 1]$.

\begin{restatable}[\textit{CATE signal from the observational data}]{proposition}{CATESignalOBS}\label{prop:obs-signal}
    Under \Cref{asmp:support,asmp:transport}, the instance-wise CATE signal $\psi_1$ in Eq.~\ref{eq:obs_signal}, which uses the outcome information from the observational data, is unbiased for the CATE in the RCT population, i.e., ${\E[\psi_1 | X] = \E[Y_1 - Y_0 | X, S = 0]}$.
\end{restatable}

We are now ready to give the core result of this section, connecting our causal assumptions to the null hypothesis of the statistical test that we will develop in the next section. 
\begin{restatable}[\textit{Null Hypothesis on Signal Difference}]{corollary}{Null Hypothesis on Signal Difference}\label{corr:obs-signal} Define $\psi = \psi_1 - \psi_0$ as the instance-wise signal difference between the observational and RCT CATE estimates. Then, under the null hypothesis, i.e. under \cref{asmp:rct-validity,asmp:support,asmp:transport}, we have it that $\E[\psi|X] = 0$.
\end{restatable}
\begin{proof}
If~\cref{asmp:rct-validity,asmp:support,asmp:transport} hold, then~\cref{prop:rct-signal,prop:obs-signal} imply that ${\E[\psi_0 | X] = \E[\psi_1 | X]} = \E[Y_1 - Y_0 | X, S=0]$. 
\end{proof}

\begin{remark}
 Note that by \Cref{corr:obs-signal}, violation of the conditional moment restrictions imply that one or more of our assumptions is incorrect, including the internal validity assumptions on the RCT. If we are willing to independently assume that the RCT is internally valid, then violation of the CMRs implies a violation of one or both of the internal and external validity assumptions \textit{on the observational data}.
\end{remark}

\subsection{Conditional Moment Restriction (CMR) Formulation and Maximum Moment Restriction-based (MMR) Tests}

For a practical approach to testing, we leverage the rich literature on conditional moment restriction (CMR) tests. Several examples exist of CMRs being used to express restrictions on functions of the data. One such example is using CMRs to reformulate instrumental variable (IV) regression \citep{zhang2020maximum}. %
However, to our knowledge, using CMRs to compare RCT and observational data as described in this paper has not been previously explored. 
We present the CMR-version of the null hypothesis in the following proposition: 

\begin{restatable}[\textit{Null Hypothesis, CMR}]{proposition}{nullhypocmr}\label{prop:moment}
    Under \Cref{asmp:rct-validity,asmp:support,asmp:transport}, we have a set of 
    conditional moment restrictions (CMRs) on the signal difference, $\psi$:
    \begin{equation}
    \label{eq:null-cmr}
      H_0: \E[\psi | X] = 0 \qquad P_{X}\text{-almost surely,}  
    \end{equation}
    where $P_{X}$ is the distribution of $X$ on the joint distribution of the RCT and observational study. \Cref{eq:null-cmr} implies an infinite set of unconditional moment restrictions, $\E[\psi f(X)] = 0, \forall f \in \cF$, where $\cF$ is the set of measurable functions on $\cX$.
\end{restatable}

The core part of \cref{prop:moment} is in showing how we can formulate the CMR given our assumptions, while the second part of the statement is straightforward and follows directly from the law of iterated expectation. Testing CMRs is challenging because an infinite number of equivalent \textit{unconditional} moment restrictions (UMR) must be considered. Thus, we follow a method proposed by \citet{muandet2020kernel}, where $\cF$ in \cref{prop:moment} is set to be a reproducing kernel Hilbert space (RKHS). They further show that using the \textit{maximum moment restriction} (MMR) within the unit ball of the RKHS as the test statistic fully captures the original set of CMRs and also has a closed-form expression that can be easily implemented.  However, note that here we are directly testing the CMRs, while  \citet{muandet2020kernel} consider testing hypotheses \textit{on statistical parameters that imply CMRs}, which leads to a larger set of assumptions on the parameters.  Therefore, in the following, we state the hypothesis test with respect to the MMR test statistic and the assumptions required for our use case.  A proof showing that these assumptions suffice for the properties of the test to hold is provided in Appendix~\ref{sec:app-proofs}.  This main result will hold for a particular class of kernels, which we define here:
\begin{restatable}[\textit{Integrally strictly positive definite (ISPD)}]{definition}{Integrally strictly positive definite (ISPD)}\label{def:ISPD}
    A kernel $k(\cdot,\cdot): \cW \times \cW \rightarrow \R$ is integrally strictly positive definite if for all $f: \cW \rightarrow \R$ satisfying $0< \|f\|_2^2 < \infty$,
    \[\int_{\cW \times \cW} f(w)k(w, w')f(w') dw dw' > 0\]
\end{restatable}
Now, we are ready to give an MMR-based hypothesis test that tests the null hypothesis given in \Cref{prop:moment}: 
\begin{restatable}[\textit{Maximum Moment Restriction-based test for CATE function}]{theorem}{thmmmr}\label{thm:mmr}
     Let $\cF$ be a RKHS with reproducing kernel $k(\cdot,\cdot): \cX \times \cX\rightarrow \R$ that is ISPD, continuous and bounded. Suppose $|\E[\psi|X]| < \infty$ almost surely in $P_X$, and $\E[[\psi k(X, X') \psi']^2] < \infty$ where $(\psi', X')$ is an independent copy of $(\psi, X)$. Let $\fM^2 = \sup_{f \in \cF, ||f|| \le 1}(\E[\psi f(X)])^2$.  Then,
     \begin{enumerate}
         \item The conditional moment testing problem in Eq. \ref{eq:null-cmr} can be reformulated in terms of the MMR as $H_0^': \fM^2 = 0$, $H_1^': \fM^2 \ne 0$.
     \end{enumerate}
     Further, let the test statistic be the empirical estimate of $\fM^2$,
     \[\hat{\fM}_n^2 = \frac{1}{n(n-1)}\sum_{i,j \in \cI, i \ne j} \psi_i k(x_i, x_j)\psi_j\]
     \begin{enumerate}
         \setcounter{enumi}{1}
         \item Then, under $H_0^'$, %
         \[ n\hat{\fM}_n^2 \xrightarrow[]{d} \sum_{j=1}^\infty \lambda_j(Z_j^2 - 1)\]
         where $Z_j$ are independent standard normal variables and $\lambda_j$ are the eigenvalues for $\psi k(x,x') \psi'$. 
         \item Under $H_1^'$,
         \[ \sqrt{n}(\hat{\fM}_n^2 - \fM^2) \xrightarrow[]{d} \cN(0, 4\sigma^2)\]
         where $\sigma^2= var_{(\psi, X)}[\E_{(\psi', X')}[\psi k(X,X') \psi']]$
     \end{enumerate}
\end{restatable}

\begin{remark}
Intuitively, the MMR test statistic is trying to find regions in $X$ where the signal difference (i.e. the difference in the CATE estimates between the observational study and RCT) is maximized. Thus, the larger the signal difference is, the larger the test statistic will be, and the more likely we will be to reject the null hypothesis. This is exactly the behavior that we want from such a test statistic. Note as well that the test statistic is computed using the signal difference directly and \textit{not} separately for each potential outcome mean. This theoretically-grounded choice follows directly from our discussion in~\cref{sec:id-results}.
\end{remark}
\begin{remark}
    Note that these asymptotic distributions imply that $n\hat{\fM}_n^2$ converges to a distribution with finite variance under the null, but diverges at a rate of $\sqrt{n}$ under the alternative hypothesis, which implies that the MMR test has asymptotic power of one.  In addition, since the null distribution does not have a closed form, to obtain the critical value for rejection, we follow Algorithm 1 in \citep{muandet2020kernel}, which uses bootstrap to simulate the null distribution.
\end{remark}
\begin{remark}
\label{re:psi}
\Cref{prop:moment} states that the true signal difference, $\psi = \psi_1 - \psi_0$, satisfies a set of CMRs. However, in practice, we perform testing using an estimate of the signal difference, $\widehat{\psi}$, where we plug-in estimates of the underlying nuisance functions, such as the propensity score, $P(A=a | S= 1, X)$. As a result, we might expect our statistical test, all else being equal, to be more likely to reject an observational study, as there are two sources of variation in the test statistic: first, in the signals themselves through the estimated nuisance functions, and second, through variation in the data that exists even when the signals are perfectly estimated. In our semi-synthetic experiments, we find that the difference in the type I error (between using $\psi$ and $\widehat{\psi}$) is minimal for moderately large sample sizes and converges to zero as the sample size increases.
\end{remark}

\subsection{Explainability of MMR-based Falsification Test}
Another appealing feature of using the MMR-based approach is that we may express the maximizer,
\begin{equation}\label{eq:witness_function}
   f^* = \arg \sup_{f \in \cF, ||f|| \le 1}(\E[\psi f(X)])^2,
\end{equation} in closed form (see the proof of \cref{cor:witness} for details). In turn, we may determine where the CATE function estimated by the RCT and the observational study is the most discrepant by looking at regions of $\cX$ with large magnitudes of $f^*$.  The function $f^*$, known as the \enquote{witness function}, can be found by the following corollary:

\begin{restatable}[]{corollary}{corwitness}
    \label{cor:witness}
    The witness function in~\cref{eq:witness_function} can be estimated as
    \[\hat{f}^*(x) = C\frac{1}{n}\sum_i \psi_i k(x_i, x)\]
    where $C$ is an unrelated constant so that $\int_{\cX} f^{*2}(x) dx = 1$.
\end{restatable}
\begin{remark}
Consider the following example where having a witness function could be beneficial: suppose an endocrinologist wants to determine whether to prescribe SGLT2-inhibitors ($A=1$) or not ($A=0$) for diabetic patients. Further assume that there is an RCT and an observational study that studies the effect of SGLT2-inhibitors on HbA1c levels ($Y$). If our MMR-based approach were to falsify the observational study, the witness function would enable the clinician to understand what types of patients (e.g. people who are $\ge 60$ years old and have history of heart disease) have conflicting conclusions in the RCT versus observational study with respect to drug benefit. 

With this information, they may seek to understand and do follow-up analyses on what violations of the causal assumptions led to the discrepancy in the \enquote{older with prior heart disease} patient population. For example, there may be a violation of internal validity of the observational data (e.g. ignorability), where there are still some unmeasured confounders in the observational study for this particular patient group. Alternatively, there could be an external validity violation (e.g. mean exchangeability of contrast), whereby the causal effects themselves are unbiased, but the standard of care of this patient population may be different between the two studies (i.e. unmeasured effect modifiers). Overall, the witness function can provide a window for clinicians to look for possible violations in a specific patient population, allowing for a richer view into observational study results.
\end{remark}

\begin{remark} 
A practitioner may interpret or visualize the witness function in a couple of different ways, which we outline here. A simple method, appropriate for domain experts (e.g. clinicians), is to use domain knowledge to pre-select a subset of covariates on which one can do low-dimensional projections for each pair. We provide an example of this method in our experimental results. Another method is to take the top or bottom $10\%$ of witness function values over $X$ and then look at characteristics of these populations. This approach can guide the choice of low-dimensional parameters to examine (e.g., for major differences across age, the witness function projected onto age can be plotted).
\end{remark}

The MMR-based testing framework is useful both because it affords a closed-form expression of the test statistic used to test the CMR in \Cref{prop:moment} \textit{and} gives an explainable view into the rejections of the test via the witness function. We argue that both are crucial for our problem of falsification of causal assumptions in observational studies, specifically internal and external validity. 
\mko{Toss in a description from the rebuttal on how explainability works}
Several other approaches exist in the literature for testing CMRs, and we point the reader to \citet{muandet2020kernel} for an overview. In the following two sections, we will tease out empirically the benefits of our testing approach against several baselines and provide an example of how the witness function can be used in practice on a real-world dataset.

\section{SEMI-SYNTHETIC EXPERIMENTS}
\label{sec:semi-synthetic}

\subsection{Setup}
 For this set of experiments, we use covariates from the Infant Health and Development Program (IHDP), an RCT run on premature infants assessing the treatment effect of professional home visits on future cognitive function \citep{brooks1992effects}. We generate an RCT and observational dataset (with simulated outcomes) from the partial IHDP dataset used in \citet{hill2011bayesian}, which contains 985 observations, 28 covariates, and one binary treatment variable. 
 
 A \enquote{simulated} dataset in our experiments consists of a single RCT and a single observational dataset. Our simulation strategy for the data draws largely on the approach taken by \citet{hussain2022falsification}. In particular, to generate the RCT, we resample the rows of the IHDP dataset to $n_0 = 2955$.
For the observational dataset, we first resample the rows of the IHDP dataset to the desired sample size, $n=s \cdot n_0$. Then, we induce a difference in the covariate distribution between the observational component and the RCT by doing weighted resampling in the observational data, such that male infants, infants whose mothers smoked, and infants with working mothers are less prevalent. To introduce explicit violations of our assumptions in the observational data, we generate $m$ confounders so that we can later conceal some of them to simulate unmeasured confounding.  Then, in both the RCT and the observational dataset, we simulate outcomes according to a response surface detailed in \cref{sec:app-semisynthetic}.  Finally, we conceal $c_z$ confounders in order of ``confounding strength'', which is determined by a vector, $\gamma \in \mathbb{R}^m$. For more information on confounder generation, outcome simulation, and bias simulation via confounder concealment, see \cref{sec:app-semisynthetic}. For parameters $m, c_z$, and $\alpha$ (significance level), we default to $m=7$, $c_z=0$, and $\alpha=0.05$ unless otherwise specified.

\subsection{Evaluation}
We evaluate our algorithm based on our original desiderata. Namely, we measure \textit{power}, i.e. the rate of rejecting the null hypothesis when the CATE function estimates from the RCT and observational study \textit{do not} converge to the same function and \textit{type I error}: rate of rejecting the null hypothesis given that they \textit{do} converge to the same function.

We use the following two baselines. \textbf{Average Treatment Effect (ATE)} – in the RCT, we compute the difference of mean outcomes between the treatment and control groups; in the observational data, we obtain an ATE estimate by leveraging recent techniques in the double machine learning (DML) and transportability literature, akin to the estimator in \citet{dahabreh2020extending}. \textbf{Group Average Treatment Effect (GATE)} – in the RCT, we compute the difference of mean outcomes between the treatment and control groups in pre-specified subgroups defined by the infant's birth weight and maternal marital status\footnote{We specify the following four subgroups: ($\ge 2000$g, married), ($< 2000$g, married), ($\ge 2000$g, single), ($< 2000$g, single)}; in the observational data, we use a transportable, doubly-robust estimator (see Appendix C in \citet{hussain2022falsification}), to estimate the GATE for each subgroup. Both baselines use hypothesis testing based on asymptotic normality of ATE or subgroup estimates. Note that this approach requires pre-specification of the subgroups.

Both baselines reflect the idea of \enquote{falsifying} the observational study by looking at a pre-specified group (subgroups, in the GATE case) to detect differences in the causal effect estimates. Our method, labeled as \textbf{MMR-Contrast}, requires no pre-specification and automatically finds \enquote{highly-discrepant} regions where the causal effect estimate is different between the RCT and the observational study.

\subsection{Results}
\textit{MMR-Contrast largely maintains the desired type I error of $0.05$ while having more power compared to baselines.} As conjectured in \cref{re:psi}, MMR-Contrast tends to slightly over-reject, which is reflected in \cref{fig:fig1-power} by the marginally elevated type I error. Furthermore, MMR-Contrast enjoys greater power than GATE and ATE, particularly in settings where the confounding bias is more subtle. We conjecture that the gain in power is due to MMR-Contrast implicitly checking across all subpopulations of $X$ for disagreements in CATE estimates. Indeed, we see that when the concealed confounder has a weight of $1$ (as opposed to $2.75$), the difference in power between MMR and ATE is much larger.

\textit{When computing MMR-Contrast with $\psi$ versus $\widehat{\psi}$, the empirical gap in type I error shrinks with increasing sample size of the observational study (see \cref{fig:psi}).} Reassuringly, we see that the level of the test is maintained at $\alpha = 0.05$ when the true signal difference, $\psi$, is used, which supports our theoretical results. Secondly, using the estimated signal difference, $\widehat{\psi}$, achieves the appropriate type I error when the observational study size at least matches the RCT, i.e. sample size ratio is 1, which one might expect in practice. 

\textit{Visualizing the witness function in \cref{fig:wf} demonstrates the covariate regions in which the observational effect estimates are increased or decreased compared to the RCT.} We largely see that the witness function yields positive values, implying that the observational study is generally estimating a larger treatment effect  (i.e. professional visits benefit child cognitive development) than the RCT.  However, there are certain subgroups, e.g. children with high birth order whose mothers do not drink and children with high neonatal health index, for which the observational study estimates lower treatment effects than the RCT.  The MMR test is able to discover these subgroup differences, leading to better power than testing for ATE or GATE. 
Another potential use case of the witness function is for development of treatment guidelines, where subgroups with high witness function values may be \enquote{flagged} as having conflicting evidence.

\begin{figure}[t!]
\begin{subfigure}[t!]{0.48\textwidth}
    \centering
    \includegraphics[width=\textwidth]{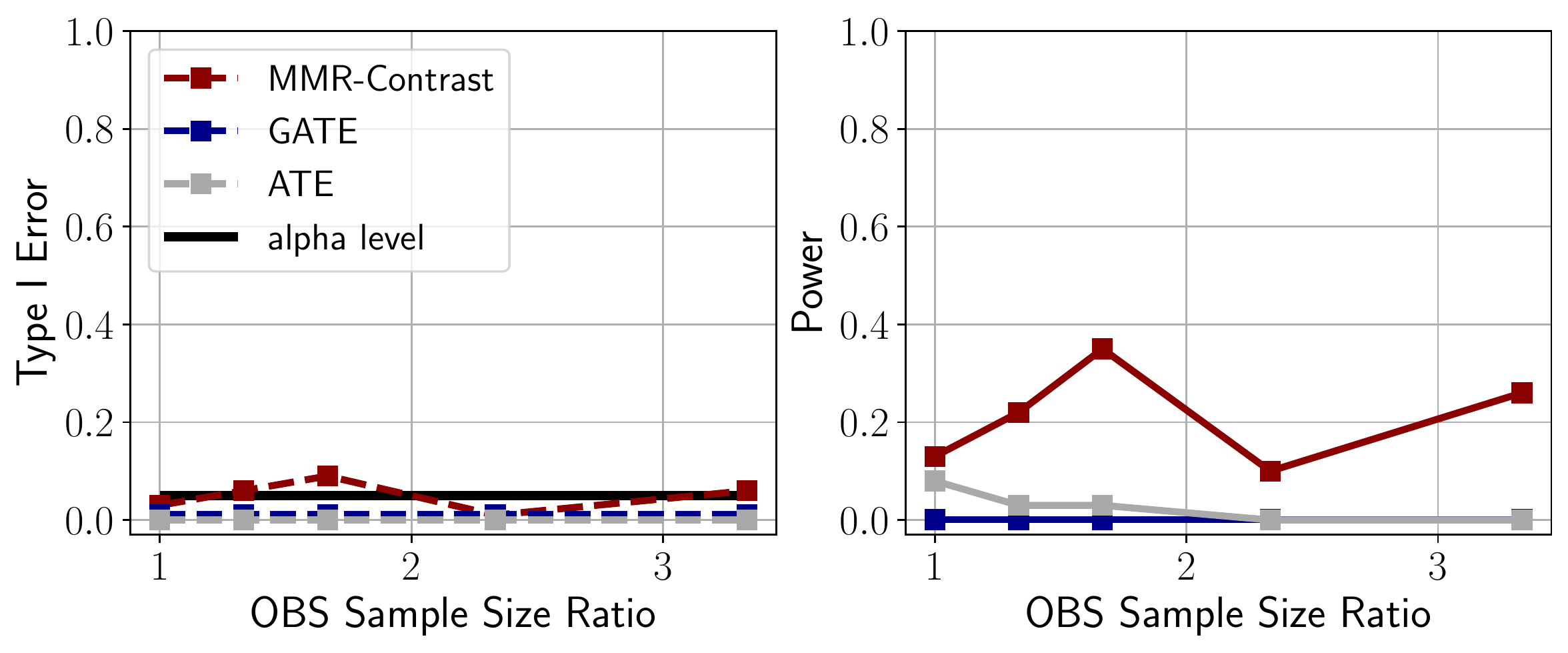}
    \caption{Low confounder strength ($\max(\gamma) = 1.$). (left) no unobserved confounders; (right): one confounder concealed}
    \label{fig:lowconf}
\end{subfigure}%
\hspace{2mm}
\medskip
\begin{subfigure}[t!]{0.48\textwidth}
    \centering
    \includegraphics[width=\textwidth]{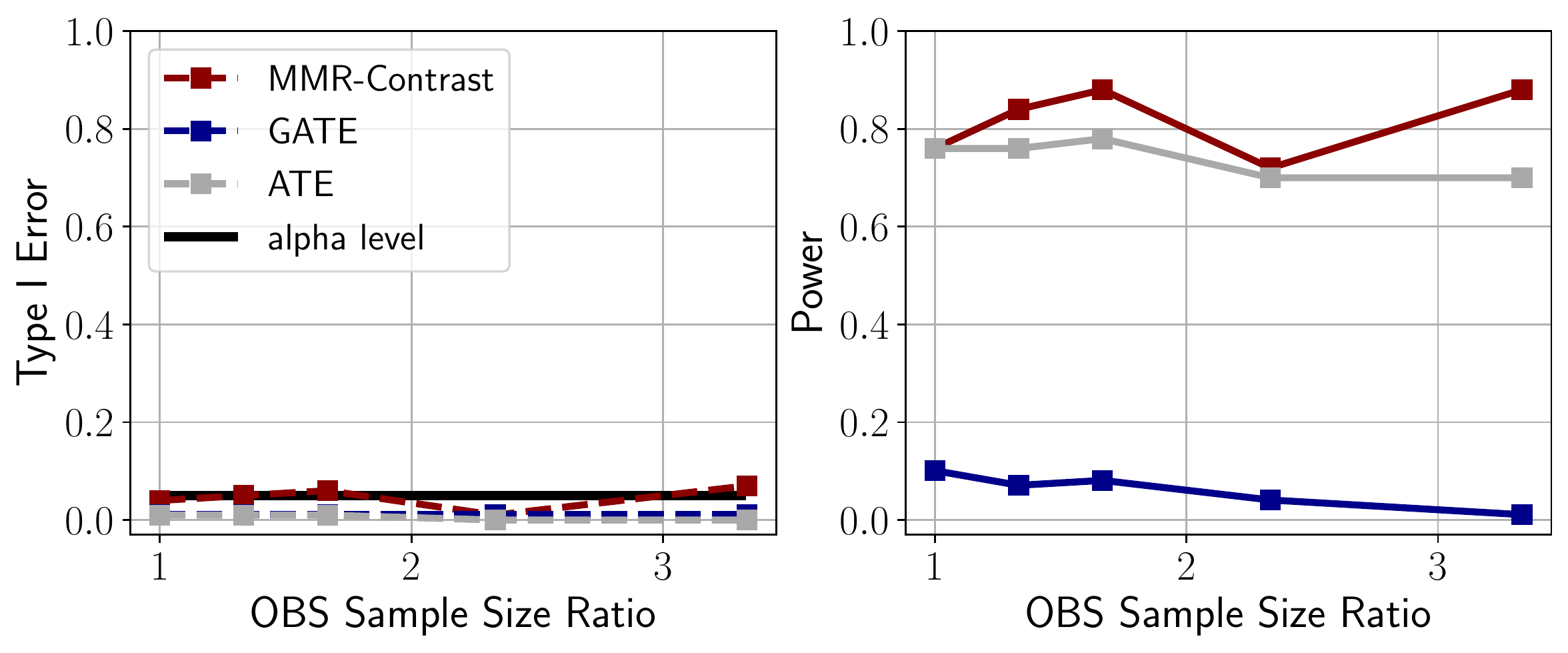}
    \caption{High confounder strength ($\max(\gamma) = 2.75$). (left) no unobserved confounders; (right): one confounder concealed}
    \label{fig:highconf}
\end{subfigure}
\caption{Type I error and power of MMR-Contrast, GATE and ATE under different confounder strengths. The left panels in \subref{fig:lowconf}) and (\subref{fig:highconf}) show that the level of all three approaches generally retains the nominal level of 0.05.  The right panels show the superior power of MMR-Contrast. Particularly, when the confounder strength is lower (as in (\subref{fig:lowconf})), the difference of CATE estimates between the observational study and the RCT is more difficult to detect, leading to a larger difference of power between MMR-Contrast and ATE. The GATE approach, since it is based on random subgroups, has minimal power, even under the high confounder strength scenario.}
\label{fig:fig1-power}
\end{figure}

\begin{figure*}[t!]
\begin{subfigure}[t!]{0.32\textwidth}
    \centering
    \includegraphics[width=\textwidth]{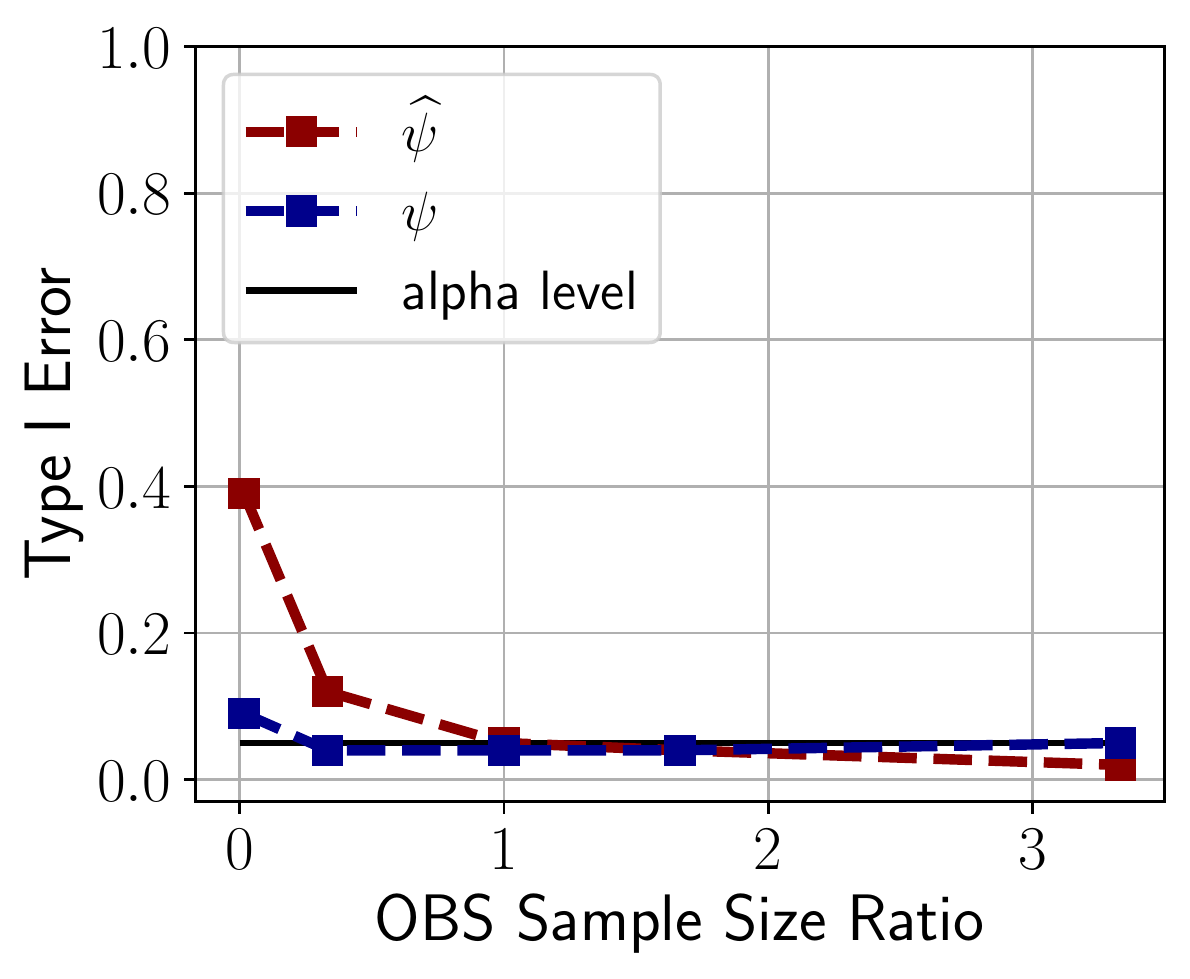}
    \caption{}
    \label{fig:psi}
\end{subfigure}%
\begin{subfigure}[t!]{0.58\textwidth}
    \centering
    \includegraphics[scale=0.45]{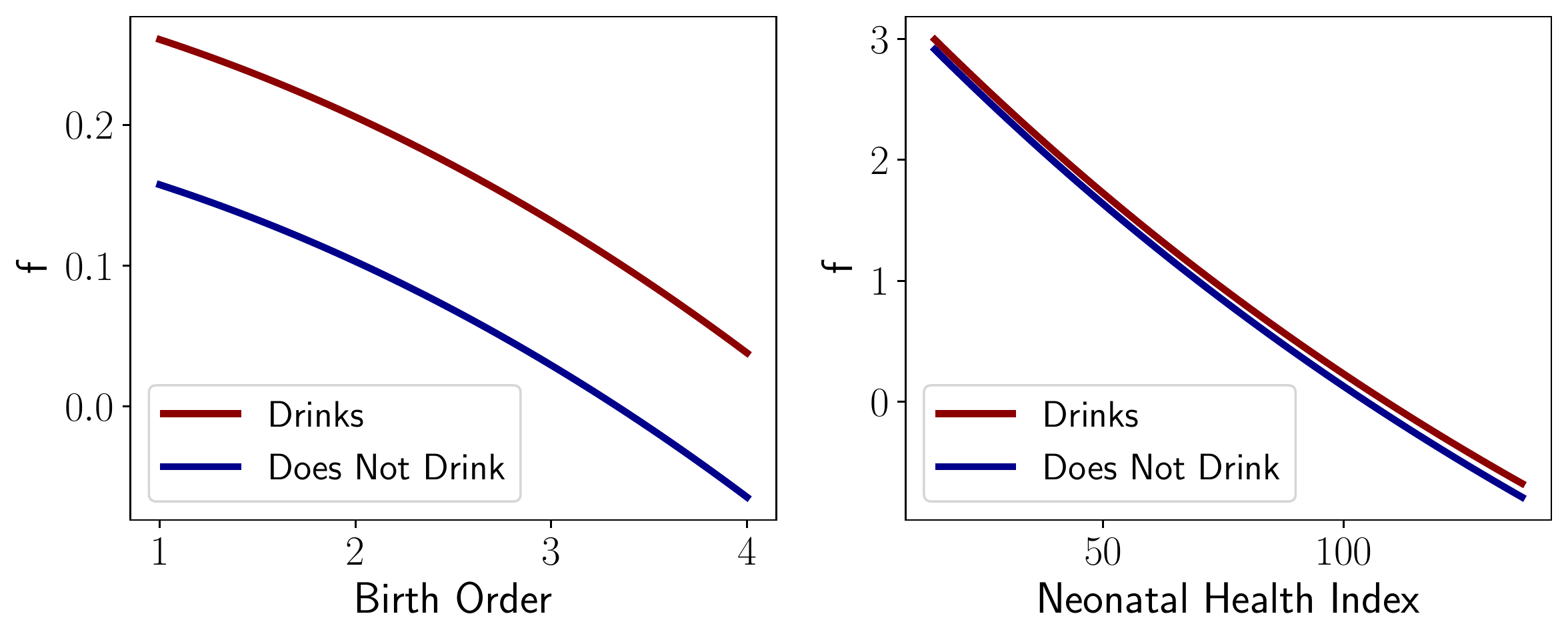}
    \caption{}
    \label{fig:wf}
\end{subfigure}
\caption{Panel (\subref{fig:lowconf}) demonstrates the relative performance of tests using test statistics computed with true signals ($\psi$) and estimated signals ($\widehat{\psi}$).  The sample size of the RCT study is fixed ($n_0 = 2955$) and the sample size ratio between the observational study and the original IHDP data ranges from 0.01 to 3.33.  The blue line shows that the test using $\psi$ achieves the nominal level (0.05).  The red line shows that under small sample sizes of the observational study, the test using $\hat{\psi}$ over-rejects due to errors in nuisance function estimation, which is consistent with our conjecture.  Nevertheless, its level promptly converges to 0.05 as the number of samples in the observational study matches or exceeds the RCT.  Panel (\subref{fig:highconf}) demonstrates the witness functions produced as a byproduct of our test, which show mostly positive values and certain negative regions.}
\label{fig:fig2-3-psiwitnessfunc}
\end{figure*}

\begin{table}
\centering
{\footnotesize
\begin{tabular}{lccc}
\toprule
\toprule
    \textit{Selection Bias}          & \textbf{MMR-Contrast}     &\textbf{ATE}     & \textbf{GATE}  \\ \midrule
\quad  $p=0$  & 0.29 & 0.32 & \textbf{0.17}  \\
\quad  $p=0.05$  & \textbf{0.67} & 0.58 & 0.40 \\
\quad  $p=0.10$  & \textbf{0.94} & 0.88 & 0.67 \\
\quad  $p=0.15$  & \textbf{1.0} & 0.98  & 0.91 \\
\bottomrule
\bottomrule
\end{tabular}}
\caption{Rejection rate when introducing different amounts of selection bias into the observational data in WHI study. $p$ stands for the strength of selection introduced in the the data (refer to Section~\ref{sec:whi} for details).}
 \label{tab:whi_table}
\end{table}

\section{WOMEN'S HEALTH INITIATIVE (WHI) EXPERIMENTS}
\label{sec:whi}
To assess our method in a practical setting, we use observational and clinical trial data from the Women's Health Initiative (WHI). These studies broadly investigate the impact of hormone therapy and vitamin D supplementation on several clinical outcomes. Conceptually, our analysis consists of first taking $B$ bootstrapped datasets from either the original WHI observational study or a \enquote{biased} version, then selecting a subset of covariates (to generate subgroups for the GATE baseline), and finally running \textit{GATE}, \textit{ATE}, and \textit{MMR-Contrast} on each bootstrap iteration. We evaluate the methods by reporting the \textit{rejection rate} in each \enquote{bias} setting. To induce different amounts of selection bias into the observational study, we drop patients who were not exposed to the intervention and did not experience the event with some probability $p$. For further details on data preprocessing, setup, and evaluation, see Appendix~\ref{sec:app-whi}. 

\subsection{Setup}
We use the Postmenopausal Hormone Therapy (PHT) trial as the RCT in our analysis, which was run on postmenopausal women aged 50-79 years with an intact uterus.  The trial investigated the effect of hormone therapy on several types of cancers, cardiovascular events, and fractures, measuring the \enquote{time-to-event} for each outcome.  In the WHI setup, the observational study component was run in parallel and tracked similar outcomes to the RCT.  Our processing of this dataset follows closely to the pre-processing steps taken by \citet{hussain2022falsification}. We binarize a composite outcome, called the \enquote{global index}, in our analysis, where $Y=1$ if coronary heart disease, stroke, pulmonary embolism, endometrial cancer, colorectal cancer, hip fracture, or death due to other causes was observed in the first seven years of follow-up, and $Y=0$ otherwise.  Note that $Y=0$ could also occur from censoring. To establish treatment and control groups in the observational study, we use questionnaire data in which participants confirm or deny usage of combination hormones (i.e. both estrogen and progesterone) in the first three years. For other covariates, we use only those measured in both the RCT and observational study to simplify the analysis. 

\subsection{Results}
\textit{MMR-Contrast has superior power compared to the baselines in real-world data.} As shown in \cref{tab:whi_table}, MMR-Contrast has the best ability to reject studies that have selection bias. Note, as well, that GATE always has lower rejection probability compared to ATE. This result implies that using the GATE approach without prior knowledge on \textit{which} subgroups lead to different effect estimates using observational versus RCT data is highly disadvantageous in terms of statistical power. Though the MMR approach is conceptually similar to GATE, it finds discrepant covariate regions in a data-driven fashion instead of requiring pre-specified groups, thus achieving better power.

\section{RELATED WORK}

\textbf{Transportability of causal effects}: A long line of work gives assumptions under which causal effect estimates can be transported from one population to another. This includes work in statistics on generalizing average effects from one or more RCTs to broader target populations \citep{Cole2010-zf, Hartman2015-td, dahabreh2019, dahabreh2020extending}, and work in computer science on giving graphical criteria for determining when effects can be transported in more general scenarios \citep{Pearl2011-rt, Pearl2014-lm, Pearl2015-ta}. \citet{Hartman2015-td} similarly consider hypothesis testing as a \enquote{placebo} test to check assumptions, though their assumptions differ from the ones we consider here. They focus on transporting effects from RCTs to a target population, and do not assume internal validity on observational data. In particular, their assumptions have a testable implication, that the average treated outcome in the target population will match the average treated outcome under a reweighting of the RCT population. Meanwhile, our focus is on testing assumptions of both transportability / external validity, as well as internal validity of observational studies.

\textbf{Combining experimental and observational data for improved estimation}: There is a recent line of work on combining observational and experimental data to yield more precise estimates of causal effects, even when the observational data may be biased \citep{rosenman2020combining, Yang2020-na, Cheng2021-sn, Chen2021-eo}. We focus on hypothesis testing as a means of falsification as our primary goal rather than merging data.
While \citet{Yang2020-na} use hypothesis testing as a part of their approach, their test depends on the parametric form of the CATE function that they seek to estimate, while our test is nonparametric in nature. 

\textbf{Minimax and variational methods for parameter estimation via CMRs}: In addition, there is a growing literature of minimax algorithms that aim to find a well-specified set of model parameters that fulfill a set of CMRs. For example, one line of work looks at constructing a minimax formulation of the generalized method of moments (GMM) framework that aims to estimate highly non-linear model parameters that are also solutions to the moment conditions implied by the problem at hand, e.g. IV regression \citep{lewis2018adversarial,dikkala2020minimax,bennett2019deep,bennett2020variational}. \citet{metzger2022adversarial} also gives asymptotic theory for minimax estimators of functionals in CMRs. Another line of work tackles the case where a set of CMRs may only weakly identify the nuisance functions, though the target parameter may still be efficiently and uniquely estimable under some conditions on the estimator \citep{kallus2021causal,bennett2022inference}.This literature focuses on the problem of parameter estimation, while we use the CMR formulation for hypothesis testing of transportability and internal validity assumptions. 

\section{DISCUSSION \& LIMITATIONS}

We have proposed a novel approach for falsifying the assumptions of observational studies using experimental data. These causal assumptions, stating the internal and external validity of observational and experimental data, imply that the conditional average treatment effect is equivalent across all observed subpopulations of $X$ in both the observational and experimental data.  This in turn gives rise to testable restrictions on the combined data distribution, implying, as we show, that the difference between two functions of the data is zero-mean for any subset of $X$.
Recent advances in the econometrics literature allow us to test such restrictions.  Our approach implicitly searches for regions of $X$ where the CATE estimates disagree between the observational and experimental data, without the need for pre-specifying these subpopulations.  Moreover, this approach yields a function that characterizes the regions where disagreement is large. Finally, we design our test to avoid rejecting studies due to differences in baseline factors that do not influence the treatment effect.

However, our approach shares certain limitations with some methods in the literature on testing for violation of causal assumptions.  In particular, while violations of the CMRs imply violations of causal assumptions, this does not directly tell us which assumptions are violated (e.g., whether the observational study is subject to hidden confounding, or whether there is simply an unobserved difference between the RCT and observational populations). Finally, due to the fact that policy guidelines can be formed from such RCTs and observational studies in society, it is important for practitioners to consider the biases in the data as well as the aforementioned limitation of our method.

\section*{Acknowledgments}
We would like to thank members of the Clinical Machine Learning group for helpful discussions and valuable feedback on the manuscript. ZH was supported by an ASPIRE award from The Mark Foundation for Cancer Research and by the National Cancer Institute of the National Institutes of Health under Award Number F30CA268631. The content is solely the responsibility of the authors and does not necessarily represent the official views of the National Institutes of Health. MCS was supported by the LEAP program from the Ministry of Science and Technology in Taiwan. MO and DS were supported in part by Office of Naval Research Award No. N00014-21- 1-2807. ID is supported by an Eric and Wendy Schmidt Center PhD fellowship. This manuscript was prepared using WHI-CTOS Research
Materials obtained from the National Heart, Lung, and Blood Institute (NHLBI) Biologic Specimen and Data Repository Information Coordinating Center and does not necessarily reflect the opinions or views of the WHI-CTOS or the NHLBI.

\bibliography{ref}
\bibliographystyle{plainnat}

\clearpage
\onecolumn
\appendix
\section*{APPENDIX}

\section{Proofs}
\label{sec:app-proofs}

\subsection{Proof of \Cref{thm:transport}} \label{sec:app-p21}
\ThmTransport*
\begin{proof}
\begin{align*}
    &\E[Y_1 - Y_0 \mid X, S = 0]  \\
    &= \E[Y_1 - Y_0 \mid X, S= 1] \\ 
    &=\E[Y_1 \mid X, S= 1] - \E[Y_0 \mid X, S= 1]  \\ 
    &=\E[Y \mid X, A=1, S= 1] - \E[Y \mid X, A=0, S= 1] 
\end{align*}
The first equality follows from the mean exchangeability of the contrast (\Cref{asmp:transport}) and the second from the linearity of the expectation operator. The final equality follows from ignorability and consistency (\Cref{asmp:support}). 
\end{proof}

\mko{
Here we show that the signal 
\begin{align}
\psi_0 &= \left(\frac{\1{A = 1}}{P(A = 1 \mid S = 0)} - \frac{\1{A = 0}}{P(A = 0 \mid S = 0)}\right) \cdot \frac{\1{S = 0}}{P(S = 0 \mid X_\tau)} Y 
\end{align}
satisfies Assumption~\ref{asmp:rct-signal} under consistency and fully randomized treatment assignment (i.e., $P(A \mid X, S = 0) = P(A \mid S = 0)$, and $Y_a \indep A$). In particular, we can observe that 
\begin{align}
    \E[\psi_0(A, Y, S) \mid X_{\tau}] &= \sum_{A, Y, S} \psi(A, Y, S) \cdot  P(A, Y, S \mid X_{\tau}) \\
    &= \sum_{A, Y, S} \frac{\1{A = 1}}{P(A = 1 \mid S = 0)} \cdot \frac{\1{S = 0}}{P(S = 0 \mid X_\tau)} Y \cdot P(A, Y, S \mid X_{\tau}) \\
    &\quad\quad - \sum_{A, Y, S} \frac{\1{A = 0}}{P(A = 0 \mid S = 0)} \cdot \frac{\1{S = 0}}{P(S = 0 \mid X_\tau)} Y \cdot P(A, Y, S \mid X_{\tau})
\end{align}
We focus on the first term, observing that the second term can be handled similarly.  We first re-write the first term as 
\begin{align}
&\sum_{A, Y, S} \frac{\1{A = 1}}{P(A = 1 \mid S = 0)} \cdot \frac{\1{S = 0}}{P(S = 0 \mid X_\tau)} Y \cdot P(Y \mid S, A, X_{\tau}) P(A \mid S, X_{\tau}) P(S \mid X_{\tau})  \\
&=\sum_{A, Y, S} \frac{\1{A = 1}}{P(A = 1 \mid S = 0)} \cdot \frac{\1{S = 0}}{\cancel{P(S = 0 \mid X_\tau)}} Y \cdot P(Y \mid S, A, X_{\tau}) P(A \mid S, X_{\tau}) \cancel{P(S \mid X_{\tau})}  \\
&=\sum_{A, Y, S} \frac{\1{A = 1}}{\cancel{P(A = 1 \mid S = 0)}} Y \cdot P(Y \mid S, A, X_{\tau}) \cancel{P(A \mid S, X_{\tau})} \1{S = 0}\\
&=\sum_{A, Y, S} Y \cdot P(Y \mid S, A, X_{\tau}) \1{S = 0, A = 1}\\
&=\sum_{A, Y, S} Y \cdot P(Y \mid S = 0, A = 1, X_{\tau})\\
&=\E[Y \mid S = 0, A = 1, X_{\tau}]
\end{align}

Note that we used the fact that $P(A = 1 \mid S = 0) = P(A = 1 \mid S = 0, X_\tau)$. The result follows from the causal assumptions, which imply that $E[Y \mid S = 0, A = 1, X_\tau] - E[Y \mid S = 0, A = 0, X_\tau] = \E[Y_1 - Y_0 \mid X_{\tau}, S = 0]$
}

\subsection{Proof of \Cref{prop:rct-signal}} \label{sec:app-prop31}
\CATESignalRCT*
\begin{proof}
Here, we show that the signal 
\begin{align*}
\psi_0 &= \left(\frac{\1{A = 1}}{P(A = 1 \mid S = 0)} - \frac{\1{A = 0}}{P(A = 0 \mid S = 0)}\right) \cdot \frac{\1{S = 0}}{P(S = 0 \mid X)} Y 
\end{align*}
is an unbiased estimator of the CATE in RCT population under consistency and fully randomized treatment assignment (i.e., $P(A \mid X, S = 0) = P(A \mid S = 0)$, and $Y_a \indep A$ as in \Cref{asmp:rct-validity}). In particular, we can observe that 
\begin{align}
    \E[\psi_0(A, Y, S, X) \mid X] &= \sum_{A, Y, S} \psi_0(A, Y, S, X) \cdot  P(A, Y, S \mid X) \nonumber \\
    &= \sum_{A, Y, S} \frac{\1{A = 1}}{P(A = 1 \mid S = 0)} \cdot \frac{\1{S = 0}}{P(S = 0 \mid X)} Y \cdot P(A, Y, S \mid X) \nonumber \\
    &\quad\quad - \sum_{A, Y, S} \frac{\1{A = 0}}{P(A = 0 \mid S = 0)} \cdot \frac{\1{S = 0}}{P(S = 0 \mid X)} Y \cdot P(A, Y, S \mid X) \label{eq:p31-1}
\end{align}
We focus on the first term, observing that the second term can be handled similarly.  We first re-write the first term as 
\begin{align}
& \sum_{A,Y,S} \frac{P(Y \mid S, A, X) P(A \mid S, X) P(S \mid X)}{P(A=1 \mid S=0) P(S = 0 \mid X)} \cdot \1{A=1, S=0} \cdot Y \nonumber \\
& =\sum_{Y} Y \cdot \frac{P(Y \mid S=0, A=1, X) \Ccancel[blue]{P(A=1 \mid S=0, X)} \Ccancel[red]{P(S=0 \mid X)}}{\Ccancel[blue]{P(A = 1 \mid S = 0)} \Ccancel[red]{P(S = 0 \mid X)}} \label{eq:p31-2} \\
& = \E[Y \mid S = 0, A = 1, X] \nonumber \\
& = \E[Y_1 \mid S = 0, A = 1, X] \nonumber \\
& = \E[Y_1 \mid X, S=0]  \nonumber 
\end{align}

Repeating the similar arguments for the second term in Eq.~\ref{eq:p31-1}, we have $\E[\psi_0(A, Y, S, X) \mid X] = \E[Y_1 - Y_0 \mid X, S = 0]$, which completes the proof.
\end{proof}

\subsection{Proof of \Cref{prop:obs-signal}} \label{sec:app-prop32}
\CATESignalOBS*

\begin{proof}
We have,
\begin{align*} 
\psi_1 = \frac{1}{P(S=0 \mid X)} \bigg[ &\1{S = 0} \left( \mu_1(X) - \mu_0(X) \right) \nonumber \\
 + &\1{S=1} \frac{P(S=0 \mid X)}{P(S=1 \mid X)} \left( \frac{\1{A=1} (Y- \mu_1(X))}{P(A=1 \mid S=1, X)} - \frac{\1{A=0}(Y-\mu_0(X))}{P(A=0 \mid S=1, X)} \right) \bigg]
\end{align*}

\begin{align}
\E[\psi_1(A,Y,S,X) \mid X] &= \frac{1}{\Ccancel[red]{P(S=0 \mid X)}} \bigg[ \sum_{A,Y} \left( \mu_1(X) - \mu_0(X) \right) P(Y \mid S=0, A, X) P(A \mid S=0, X) \Ccancel[red]{P(S=0 \mid X)} \nonumber \\
+\frac{\Ccancel[red]{P(S=0 \mid X})}{\Ccancel[green]{P(S=1 \mid X)}} \bigg( & \sum_{Y} \frac{Y- \mu_1(X)}{\Ccancel[blue]{P(A=1 \mid S=1, X)}} P(Y \mid S=1, A=1, X) \Ccancel[blue]{P(A=1 \mid S=1, X)} \Ccancel[green]{P(S=1 \mid X)} \nonumber \\
-& \sum_{Y} \frac{Y- \mu_0(X)}{\Ccancel[orange]{P(A=0 \mid S=1, X)}} P(Y \mid S=1, A=0, X) \Ccancel[orange]{P(A=0 \mid S=1, X)} \Ccancel[green]{P(S=1 \mid X)} \bigg) \bigg] \nonumber \\
& =\sum_{A,Y} \left( \mu_1(X) - \mu_0(X) \right) P(Y \mid S=0, A, X) P(A \mid S=0, X) \nonumber \\
& \hspace{10pt} + \sum_{Y} \left( Y- \mu_1 (X) \right) P(Y \mid S=1, A=1, X) \nonumber  \\
& \hspace{10pt} - \sum_{Y} \left( Y- \mu_0 (X) \right) P(Y \mid S=1, A=0, X) \nonumber \\
& = \left( \mu_1(X) - \mu_0(X) \right) \underbrace{\sum_{A,Y} P(Y, A \mid S=0, X)}_{=1} \nonumber \\
& \hspace{10pt} + \E[Y \mid S=1, A=1, X] - \mu_1 (X) \cdot \underbrace{\sum_{Y} P(Y \mid S=1, A=1, X)}_{=1} \nonumber  \\
& \hspace{10pt} - \E[Y \mid S=1, A=0, X] - \mu_0 (X) \cdot \underbrace{\sum_{Y} P(Y \mid S=1, A=0, X)}_{=1}  \label{eq:p32-1} \\
& = \E[Y \mid S=1, A=1, X] - \E[Y \mid S=1, A=0, X] \label{eq:p32-2} \\
& = \E[Y_1 - Y_0 \mid X, S = 0] \label{eq:p32-3}
\end{align}
Note that we have Eq.~\ref{eq:p32-1} and Eq.~\ref{eq:p32-2} since $\mu_a(X) = \E[Y \mid S=1, A=a, X]$. Eq.~\ref{eq:p32-3} follows from \cref{thm:transport}.
\end{proof}

\subsection{Proof of \Cref{prop:moment}} \label{sec:app-p33}
We restate \Cref{prop:moment} here for convenience. 
\nullhypocmr*

\begin{proof}
Under \Cref{asmp:rct-validity,asmp:support,asmp:transport}, we have $\E[\psi_0 | X] = \E[Y_1 - Y_0 | X, S = 0]$ and $\E[\psi_1 | X] = \E[Y_1 - Y_0 | X, S = 0]$ by \Cref{prop:rct-signal,prop:obs-signal} as discussed in Section~\ref{sec:cate-est}. That is, $\E[\psi | X] = 0$ where $\psi = \psi_1 = \psi_0$ is the signal difference given two CATE signals. Let $\cF$ be the set of measurable functions on $\cX$. Then, by the Law of Iterated Expectations we have
\begin{equation*}
    \E[\psi f(X)] = \E_{X}[\E[\psi f(X) | X]] = \E_{X}[\E[\psi | X] f(X )], \qquad P_{X}\text{-a.s.,}~\forall f \in \cF    
\end{equation*}
We see that Eq.~\ref{eq:null-cmr} implies the following infinite set of unconditional moment restrictions, 
\begin{equation*}
\E[\psi f(X)] = 0, \qquad P_{X}\text{-a.s.,}~\forall f \in \cF    
\end{equation*}
\end{proof}

\ilkrm{
\begin{proof}
First, to show how to express our null hypothesis as a CMR, we have that, under \Cref{asmp:rct-signal}, $\E[\psi_0 | X_\tau] = \E[Y_1 - Y_0 | X_\tau, S = 0]$, and under \Cref{asmp:support,asmp:transport,asmp:obs-signal}, we have $\E[\psi_1 | X_\tau] = \E[Y_1 - Y_0 | X_\tau, S = 0]$. Under the null hypothesis, we have it that the CATE inferred from the RCT and the CATE inferred from the observational study transported to the RCT population are equivalent, i.e. $\E[\psi_0 | X_\tau] = \E[\psi_1 | X_\tau]$. Rearranging terms and from our definition of $\psi$, we arrive at the CMR formulation of the null hypothesis. 

Now, suppose we have a measurable set of functions, $\cF$, on $\cX_\tau$. Then, we have, $\E[\psi f(X_\tau)] = \E_{X_\tau}[\E[\psi f(X_\tau) | X_\tau]] = \E_{X_\tau}[\E[\psi | X_\tau] f(X_\tau )]$ for any $f \in \cF$, which follows from the Law of Iterated Expectation. We see that Eq.~\ref{eq:null-cmr} implies the following infinite set of unconditional moment restrictions, 
\begin{equation*}
\E[\psi f(X_\tau)] = 0, \quad \forall f \in \cF    
\end{equation*}

\end{proof}
}

\subsection{Proofs for Theorem~\ref{thm:mmr} and Corollary~\ref{cor:witness}}

We restate the theorem and corollary here for convenience.
\thmmmr*
\corwitness*
The following proof follows \citep{muandet2020kernel}.  Let us define the following operator,
\begin{equation} \label{eq:m-operator}
    M f = \E [\psi f(X)]
\end{equation}
where $f \in \cF$.  Since $|\E[\psi|X]| < \infty$ almost surely in $P_X$, $M$ is a bounded linear operator. By Riesz representation theorem, there exists a unique $g \in \cF$ such that
\begin{equation*}
   M f = \langle f,g \rangle
\end{equation*}
where 
\begin{equation*}
g = \E[\psi k(X,\cdot)].
\end{equation*}
$g$ is called the conditional moment embedding (CMME) of the CMR, $\E[\psi|X]$, in $\cF$ w.r.t. $P_X$. Therefore, it follows that
\begin{equation*}
        \fM^2 = \sup_{f \in \cF, \|f\| \leq 1} \left(\E[\psi f(X)]\right)^2 = \sup_{f \in \cF, \|f\| \leq 1} \langle f, g \rangle^2 = \left\langle \frac{g}{\|g\|}, g \right\rangle^2 = \|g\|^2
    \end{equation*}

Note that the above implies that the witness function $f^* = \arg\sup_{f \in \cF, \|f\|\le 1} (\E[\psi f(X)])^2 = \frac{g}{\|g\|}$. Since $g$ is defined as $\E[\psi k(X, .)]$, it can be empirically estimated as $\frac{1}{n}\sum_{i=1}^n \psi_i k(x_i, .)$, which leads to Corollary~\ref{cor:witness}.

Since $\fM^2 = \|g\|^2$, the first statement in Theorem~\ref{thm:mmr} is essentially \[\E[\psi | X] = 0, P_X\text{-almost surely} \Leftrightarrow \|g\|^2 = 0\]
That is, $g \in \cF$ fully captures the information of the CMR for all $x \in \cX$.  This equivalence, which we will now prove, is crucial since our statistical test is based on $\|g\|^2$ and its estimates, while Proposition~\ref{prop:moment} is directed to the CMR:

($\Rightarrow$) We note that since $\cF$ is a Hilbert space, it follows that $g \in \cF$, and $\forall f \in \cF, \langle f,g \rangle = \E[\psi f(X)] = 0$ (Proposition~\ref{prop:moment}). $g$ can now only be a zero vector.  Therefore, $\|g\|^2 = 0$.

($\Leftarrow$) 
\begin{align*}
    &\|g\|^2 = 0\\
    \Rightarrow~ & \left\|\E[\psi k(X,.)]\right\|^2 = 0 \\
    \Rightarrow~ & \left\|\E[\E[\psi|X]k(X,.)]\right\|^2 = 0\\
    \Rightarrow~ & \left\|\int_{\cX}k(x,.)\E[\psi|x] p_{X}(x) dx\right\|^2 = 0\\
    \Rightarrow~ & \iint_{\cX \times \cX} p_{X}(x)\E[\psi|x]k(x,x^')\E[\psi|x']p_{X}(x^') dx dx' = 0\\
    \Rightarrow~ & \|\E[\psi|x]p_{X}(x)\|^2 = 0 \qquad (\because k(\cdot,\cdot)\text{ is ISPD})\\
    \Rightarrow~ & \E[\psi|x] = 0, P_{X}\text{-almost surely}\\
\end{align*}

Finally, we move to the second and third statements of Theorem~\ref{thm:mmr}, which define the estimator and its statistical properties.  Since $\fM^2 = \|g\|^2 = \|\E[\psi k(X,. )]\|^2 = \E[\E[\psi k(X,X') \psi']]$ where $(X, \psi)$ and $(X', \psi')$ are independently and identically distributed, we may use a $U$-statistic to estimate $\fM^2$, which is exactly
\[\hat{\fM}_n^2 = \frac{1}{n(n-1)}\sum_{i,j \in \cI, i \ne j} \psi_i k(x_i, x_j)\psi_j\]
The asymptotic distribution of $U$-statistics has been investigated intensively in the literature.  Specifically, from Section 5.5 of \citet{Serfling2009-th}, taking the special case of kernels with two inputs, we have the following lemma:
\begin{restatable}[(Serfling)]{lemma}{Serfling}
    \label{lem:Serfling}
    Given a kernel $h(.,.): \cW\times\cW \rightarrow \R$ where $\E_{(W,W')}[h(W, W')] = \theta$ and $\E_{(W, W')}[h^2(W, W')] < \infty$, the asymptotic distribution of the $U$-statistic $U_n = \frac{1}{n(n-1)}\sum_{i \ne j}h(w_i,w_j)$ can be categorized into two cases based on $\zeta_1 = var_{W}(\E_{W'}[h(W, W')])$:
    \begin{centermath}
        \left\{
        \begin{array}{lr}
            \sqrt{n}(U_n-\theta) \xrightarrow[]{d} N(0, 4\zeta_1), & \zeta_1 > 0\\
            n(U_n-\theta) \xrightarrow[]{d} \sum_{j=1}^\infty \lambda_j (Z_j^2-1), & \zeta_1 = 0
        \end{array}
        \right.
    \end{centermath}
    where $Z_j$ are independent standard normal variables and $\lambda_j$ are the eigenvalues of $h$, i.e. the solutions for $\E_{W'}[h(W',w)v(w)]-\lambda v(w) = 0$
\end{restatable}
Note that if we set $W = (\psi, X)$, $h(W,W') = \psi k(X,X') \psi'$, $\theta = \fM^2$, $\zeta_1 = \sigma^2$, the second and third statements of Theorem~\ref{thm:mmr} holds as long as $\fM^2 = 0 \Leftrightarrow \sigma^2 = var_{(\psi, X)}[\E_{(\psi', X')}[\psi k(X,X')\psi']] = 0$, which we will now show:

($\Rightarrow$)
\[\E_{(\psi', X')}[\psi k(X,X')\psi'] = \langle \psi k(X,.) , \E_{(\psi', X')}[\psi' k(X',.)]\rangle = \|\psi k(X,.)\| \left\langle \frac{\psi k(X,.)}{\|\psi k(X,.)\|} , g\right\rangle\]
Now since 
\[\frac{\psi k(X,.)}{\|\psi k(X,.)\|} \in \cF, \left\|\frac{\psi k(X,.)}{\|\psi k(X,.)\|}\right\| = 1\] 
and 
\[\fM^2 = 0 \Rightarrow \sup_{f\in\cF, \|f\| \le 1}\langle f, g\rangle = 0 \Rightarrow \langle f, g\rangle = 0, \forall f\in\cF, \|f\| \le 1\]
We conclude $\fM^2 = 0 \Rightarrow \left\langle \frac{\psi k(X,.)}{\|\psi k(X,.)\|} , g \right\rangle = 0 \Rightarrow \E_{(\psi', X')}[\psi k(X,X')\psi'] = 0 \Rightarrow var_{(\psi, X)}[\E_{(\psi', X')}[\psi k(X,X')\psi']] = 0$

($\Leftarrow$)

We first note that $var_{(\psi, X)}(\E_{(\psi', X')}[\psi k(X,X')\psi']) = 0$ implies that $\E_{(\psi', X')}[\psi k(X,X')\psi']$ is a constant $P_{(\psi,X)}$-almost surely.  We denote this constant as $c$ so we have
\begin{equation}
    \E_{(\psi', X')}[\psi k(X,X')\psi'] = c, P_{(\psi,X)}\text{-almost surely}
    \label{eq:MMR-leftarrow}
\end{equation}
From the definition of $\psi$, let $X = x^*$ be in the support of the observational study, then
\begin{align*}
    \E[\psi| S=1, X=x^*] &= \frac{1}{P(S=1|X=x^*)}\E\bigg[\frac{\mathbf{1}(A=1)(Y-\mu_1(x^*))}{P(A=1|S=1,X=x^*)}\\
    &\qquad\qquad-\frac{\mathbf{1}(A=0)(Y-\mu_0(x^*))}{P(A=0|S=1,X=x^*)}|S=1,X=x^*\bigg]\\
    &=\frac{1}{P(S=1|X=x^*)}\bigg[\E[Y-\mu_1(x^*)|A=1, S=1, X=x^*] \\
    &\qquad\qquad- \E[Y-\mu_0(x^*)|A=0, S=1, X=x^*]\bigg]\\
    &=0,
\end{align*}
where the last equality stems from the definition of $\mu_1$ and $\mu_0$.  Now note that 
\begin{align*}
  \E_\psi[\E_{(\psi', X')}[\psi k(X, X')\psi'|S=1, X = x^*]] &= \E_\psi[\E_{(\psi', X')}[\psi k(x^*, X') \psi'|S=1, X = x^*]]\\
  &= \E_\psi[\psi| S=1, X=x^*]\E_{(\psi', X')}[k(x^*, X') \psi'] \\
  &= 0 \cdot \E_{(\psi', X')}[k(x^*, X') \psi'] = 0
\end{align*}
But also we have, from (\ref{eq:MMR-leftarrow}),
\[\E_\psi[\E_{(\psi', X')}[\psi k(X, X')\psi'|S=1, X = x^*]] = \E_\psi[c] = c\]
Therefore, we have $c = 0$ and thus
\[\fM^2 = \E_{(\psi, X)}[\E_{(\psi', X')}[\psi k(X, X')\psi']] = \E_{(\psi, X)}[0] = 0\]
so this side of the arrow is also proven.

\section{Motivating Empirical Examples of Generalization of Treatment Effects rather than Counterfactual Means}
\label{sec:app-contrast}

In~\cref{fig:comparisons} we plot data that is publicly available in~\citet{sprint2015randomized} and~\citet{franklin2021emulating}.  The former is a randomized trial that reports on outcomes across subgroups, where we observe that subgroups often have larger differences in their baseline outcomes than in their treatment effects.  The latter is a study that attempts to replicate ten RCTs using observational data.  For each observational study and trial, they report on not only the resulting differences in rates (between treatment and control), but also the marginal rates under each of treatment and control.  This is done for both the observational studies and the original RCTs. We can observe that the estimated \enquote{treatment effects} tend to be closer together (between the observational studies and RCTs) than the estimated \enquote{counterfactual means}, such as the marginal rate under control. We can view this empirical example as one where \Cref{asmp:transport} approximately holds in practice, i.e. the treatment effect appears to generalize across observational and RCT populations, but the counterfactual means do not. 

\begin{figure*}[t!]
\begin{subfigure}[t!]{0.9\textwidth}
    \centering
    \includegraphics[width=0.9\textwidth]{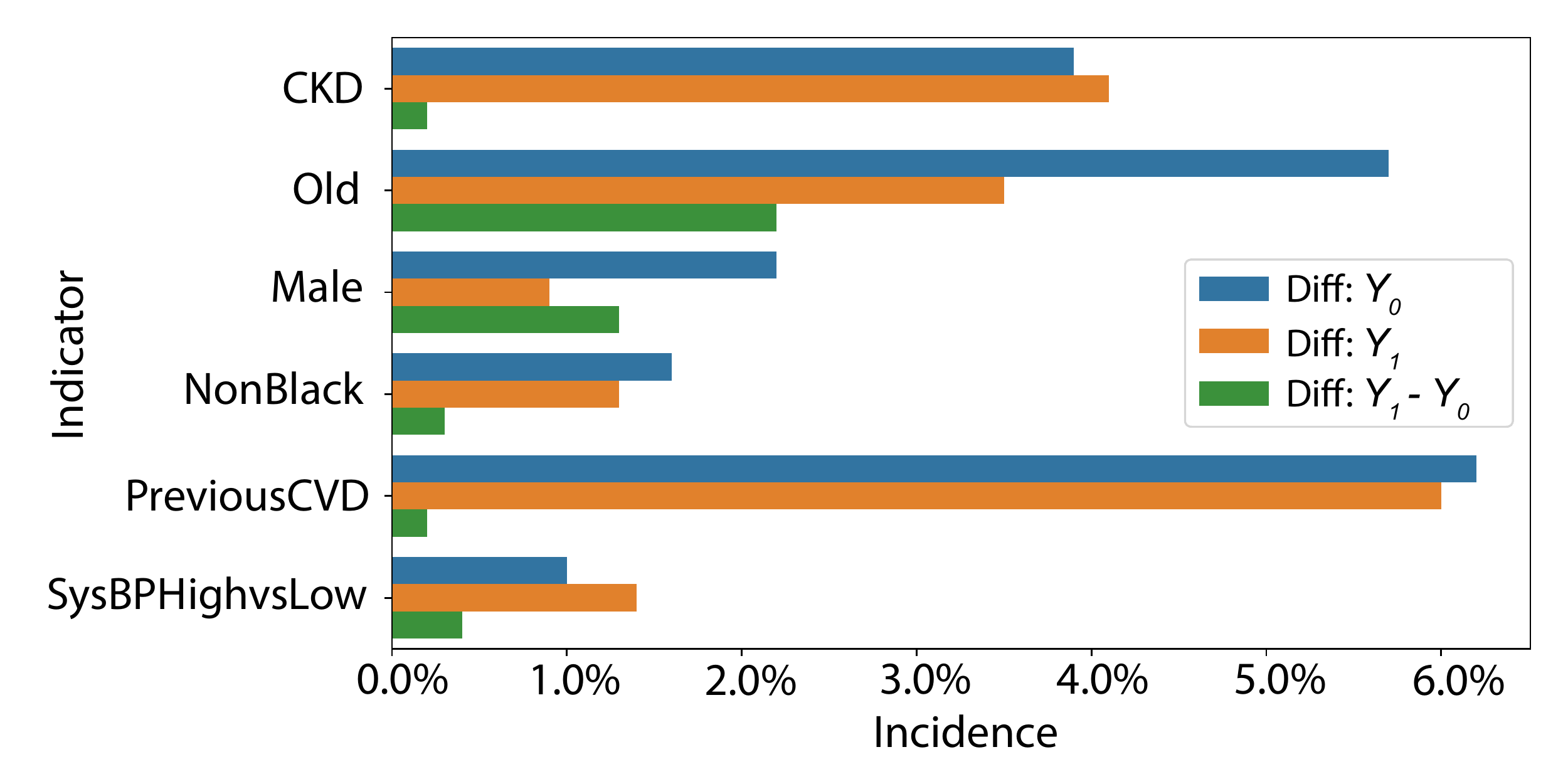}
    \caption{}
    \label{fig:sprint}
\end{subfigure}
\begin{subfigure}[t!]{0.9\textwidth}
    \centering
    \includegraphics[width=0.9\textwidth]{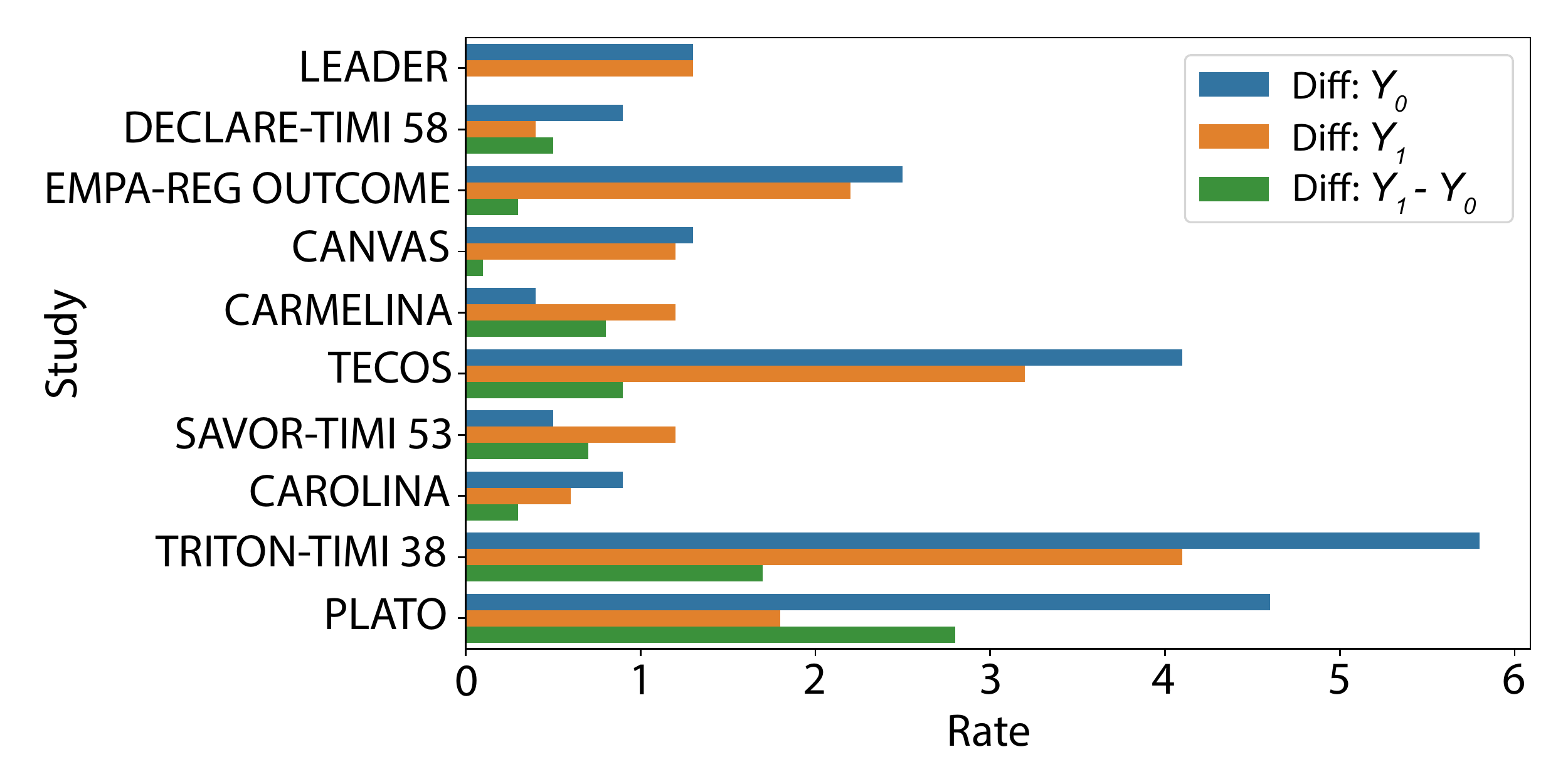}
    \caption{}
    \label{fig:rctduplicate}
\end{subfigure}
\caption{(\subref{fig:sprint}) For each binary group indicator $I$ in the SPRINT Trial, we compare the absolute difference between $\E[Y_0 \mid I = 1]$ versus $\E[Y_0 \mid I = 0]$, and similarly for $Y_1$ and $Y_1 - Y_0$, where $Y$ is a binary variable indicating the observation of the primary composite outcome.  The data supporting this plot is taken from Figure 4 of \citet{sprint2015randomized}.  Generally, the latter difference is smaller (sometimes by an order of magnitude) than the differences for individual potential outcomes. (\subref{fig:rctduplicate}) For each attempted replication of an RCT by an Observational study, we compare the differences in the reported incidence rates under treatment, under the control/comparator, and the difference between the two (analogous to the treatment effect). The latter tends to be smaller than both of the differences in counterfactual means in 6 / 10 replications, and smaller than at least one of the differences in counterfactual means in all 10 replications.  This data is taken from Table 2 of \citet{franklin2021emulating}, where we use the reported statistics in Table 2.}
\label{fig:comparisons}
\end{figure*}

\section{IHDP Experiment Details}
\label{sec:app-semisynthetic}
For both the semi-synthetic and real-world experiments, we follow closely the setup proposed by \citet{hussain2022falsification}, with a few differences highlighted below. 

\subsection{Confounder Generation \& Outcome Simulation}
We generate one RCT and one observational study in each of our 100 simulations, with the randomness appearing in our confounder generation, simulation of the potential outcomes, and the amount of noise in each. In the RCT, we retain the original IHDP data, i.e. the covariates and the binary treatment variable, but resample the dataset with equal probability to generate a final dataset of size, $n_0 = 2955$. For generation of the observational dataset, we first resample the rows of the IHDP dataset to the desired sample size, $n=s\cdot n_0$, but do the resampling in a weighted fashion, such that male infants, infants whose mothers smoked, and infants with working mothers are less prevalent. The weights are set as, 
\begin{align*}
    w = \frac{1}{1 + \text{exp}(-0.2(\mathbf{1}(\text{male infant}) + \mathbf{1}(\text{mother smoked}) + \mathbf{1}(\text{mother worked during pregnancy}) ))}
\end{align*}
Note that this differs from the reweighting scheme used in \citet{hussain2022falsification} in that we use a non-linearity in the reweighting, since we wish for the covariates used in the reweighting (i.e. sex, smoking status, working status) to be effect modifiers.

Next, we generate confounders for the observational dataset. Each confounder, $z$, is a function of a subset of the covariates, $X_s$, and the treatment, $A$: 
\begin{align*}
    z = X_s^{\top} \xi + X_s^{\top} \delta \odot A + \mathcal{N}(0,1),
\end{align*}
where $X_s \in \mathbb{R}^{4}$ and the coefficients $\xi$ and $\delta$ are set as: $\xi = (0.1, -0.1, 0.2, -0.3, 0.4)$ and $\delta = (1.,-.1,.5,-3,4)$. $X_s$ consists of the following covariates — (\enquote{neonatal health index}, \enquote{birth order of infant}, \enquote{drinks alcohol or not}, \enquote{mother finished high school}). Confounder generation for the RCT is similar but does not include any dependency on the treatment: $X_s^\top \xi + \mathcal{N}(0,1)$. We repeat this procedure $m$ times to yield $m$ confounders. 

To detail the outcome simulation, we borrow notation from \citet{hussain2022falsification}, where we let $Z\in \mathbb{R}^m$ denote the generated confounder vector and $X\in \mathbb{R}^{m_x}$ denote the covariate vector, where $m_x=28$ is the number of covariates in the original IHDP dataset. Similarly, we let $\tilde{X}=(A,X^\top)^\top$. Then, we set the following counterfactual outcome distributions: 
\begin{align*}
    Y_0 &\sim \mathcal{N}\left(\left(\tilde{X}+\frac{1}{2}\mathbf{1}\right)^\top\beta + Z^\top \gamma, 1\right) \\ 
    Y_1 &\sim \mathcal{N}(\tilde{X}^\top\beta + Z^\top\delta + \omega, 1),
\end{align*}
where $\mathbf{1}\in \mathbb{R}^{m_x+1}$ is a vector of ones, $\beta \in \mathbb{R}^{m_x+1}$ is a vector where each element is randomly sampled from $(0,0.1,0.2,0.3,0.4)$ with probabilities $(0.6,0.1,0.1,0.1,0.1)$, and $\gamma \in \mathbb{R}^m$ is a vector where each element is randomly sampled from one of two vectors with uniform probability depending on the strength of confounding desired: $(0.1,0.2,.5,.75,1.)$ or $(1.,1.75,2.,2.25,2.75)$. In \Cref{fig:lowconf}, for example, we sample from the first vector to generate the confounders, and in \Cref{fig:highconf}, we sample from the second vector. The observed outcome is then set as, $Y := AY_1 + (1-A)Y_0$. Finally, $\omega=23$ to bound the magnitude of the counterfactual outcome under treatment. We conceal confounders in order to simulate unobserved confounding, letting $c_z$ be the number of confounders concealed. As alluded to in the main paper, the order of how we conceal confounders is determined by their \enquote{confounding strength}, i.e. from highest to lowest weighted.

\section{Women's Health Initiative (WHI) Experiment Details}
\label{sec:app-whi}
We follow substantially the same setup as in \citet{hussain2022falsification}, which we recounted partially in the main paper (\Cref{sec:whi}) but do so fully in this section. The WHI conducted several clinical trials as well as an observational study in parallel to study the effect of various hormonal and dietary interventions on the health and quality of life of postmenopausal women. As mentioned in the main paper, of the three clinical trials run by WHI, we use the Postmenopausal Hormone Therapy (PHT) trial for our analysis, which looked at the effect of combination hormone therapy on postmenopausal women aged 50-79 years who had not undergone a hysterectomy \citep{rossouw2002risks}. The data used in our analysis is publicly available on BIOLINCC (\url{https://biolincc.nhlbi.nih.gov/studies/whi_ctos}). 

\subsection{Data} 

We briefly review the core characteristics of both the RCT and observational study components of the WHI study. The RCT studies the effect of a combination of 2.5mg of medoxyprogesterone and 0.625mg of estrogen on a population of $N_{HT} = 16608$ postmenopausal women. Each patient is randomly assigned to either the treatment group (i.e. estrogen + progesterone combination is given) or the control group, in which the placebo is given. The outcomes tracked in the RCT are of three categories: 1) cardiovascular events, including coronary heart disease, 2) cancers (endometrial, breast, etc.), and 3) fractures (e.g. hip, bone, etc.).  
The observational study component studies similar outcomes in a cohort of $N = 93676$ women. Women were recruited for this component in 1996, and follow-up was done until 2005, which is a similar timeframe as the RCT. Information about therapies that the patients were taking across the follow-up were tracked via questionnaires, which were taken on a yearly basis. 

\subsection{Outcome and Intervention}

As mentioned in \Cref{sec:whi} of the main paper, we define a binary outcome based on the \enquote{global index} score given to each patient, which is a composite index derived from whether or not a patient experiences any one of the following events: coronary heart disease, stroke, pulmonary embolism, endometrial cancer, colorectal cancer, hip fracture, or death due to other causes. Furthermore, the we let $Y=1$ if any one of these events is observed in the first seven years of follow-up and $Y=0$ otherwise. Notably, $Y=0$ may also occur due to censoring. 

In terms of the intervention, the RCT is run as an \enquote{intention-to-treat} trial. For the observational study component, we determine treatment and control groups based on explicit affirmation or denial of the use of estrogen and progesterone combination therapy in the first three years, which we glean from the annual survey data. Using this procedure, we end up with a total of $N_{OS} = 33511$ patients. Finally, we restrict the set of covariates used to those that are measured in both the RCT and the observational study. Each covariate indicates the same meaning, since the same set of questionnaires are used to gather them. The resulting number of covariates is $1576$. 

\subsection{Experimental Workflow}
We detail our experimental setup in this section. We note the following algorithm applies to one row of \Cref{tab:whi_table} in the main paper. Indeed, to get the remaining results, we re-apply this algorithm after \enquote{introducing} some selection bias into the observational dataset. We have the following experimental workflow: 
\begin{itemize}
    \item Step 1: Generate $B$ bootstrapped datasets of the base WHI observational dataset. 
    \item Step 2: Set list of $r$ covariate pairs $X_1,\ldots,X_r$. We use the same set of covariates as used by \citet{hussain2022falsification}, which can be found in Appendix E of their paper, to generate the $r$ covariate pairs. 
    \item Step 3: For $i = 1\rightarrow B$ 
    \begin{itemize}
        \item Apply \textbf{MMR-Contrast} (see Appendix \ref{sec:app-imp} for implementation details). Set $\Lambda_{MMR}[i] = 1$ if p-value is $<0.05$, else let $\Lambda_{MMR}[i] = 0$.
        \item Apply \textbf{ATE}. We use the same estimator as \citet{hussain2022falsification}, but average over the entire population to get the ATE from the observational study and RCT, respectively (i.e. set each patient to be part of the same group). Set $\Lambda_{ATE}[i] = 1$ if the test rejects the null hypothesis and $\Lambda_{ATE}[i] = 0$ otherwise. 
        \item For $j = 1 \rightarrow r$ 
        \item \quad Apply \textbf{GATE} using the four subgroups derived from $X_j$, as in \citet{hussain2022falsification}. Set $\Lambda_{GATE}[i][j] = 1$ if the test rejects the null hypothesis and $\Lambda_{GATE}[i][j] = 0$ otherwise. 
    \end{itemize}
\end{itemize}
Thus, the rejection rates for \textbf{MMR-Contrast}, \textbf{ATE}, \textbf{GATE} are $\frac{1}{B}\sum_i \Lambda_{MMR}[i]$, $\frac{1}{B}\sum_i \Lambda_{ATE}[i]$, $\frac{1}{B\cdot r}\sum_i\sum_j \Lambda_{GATE}[i][j]$, respectively. We repeat the above workflow to the WHI dataset that has induced selection bias. To add selection bias to the data, with probability $p$, we drop patients who were not exposed to the intervention and did not experience the event. To obtain the results for \Cref{tab:whi_table}, we run our experimental procedure for $p=(0.,0.05,0.10,0.15)$. 

\section{Details on Implementation of the MMR-Contrast Method}
\label{sec:app-imp}
The implementation of the MMR-Contrast method references the workflow illustrated in  \citep{hussain2022falsification} and \citep{muandet2020kernel}: the former for signal calculation and the latter for significance testing.

\subsection{Calculation of signal difference}
As elaborated in the main text, in the combined data (combining the RCT and observational study), for each observation $i$ producing data $(y_i, s_i, a_i, x_i)$, we define the true signal difference as,
\begin{align*}
    &\psi_i = \frac{1}{P(S=0|X=x_i)}\left\{\mathbf{1}(s_i = 0)\left[\left[\mu_1(x_i)-\mu_0(x_i)\right]-\left[\frac{\mathbf{1}(a_i=1)}{P(A=1|S=0)}-\frac{\mathbf{1}(a_i=0)}{P(A=0|S=0)}\right]y_i\right]+\right.\\
    & \qquad\qquad\qquad\qquad\qquad\qquad \left.\mathbf{1}(s_i = 1)\frac{P(S=0|X=x_i)}{P(S=1|X=x_i)}\left[\frac{\mathbf{1}(a_i=1)(y_i - \mu_1(x_i))}{P(A=1|S=1,X=x_i)}-\frac{\mathbf{1}(a_i=0)(y_i - \mu_0(x_i))}{P(A=0|S=1,X=x_i)}\right]\right\}    
\end{align*}
where $\mu_1(x_i) = \E[Y|S=0,A=1,X=x_i]$ and $\mu_0(x_i) = \E[Y|S=0,A=0,X=x_i]$.  Note that the true signal difference includes several unknown nuisance functions that need to be estimated:
\begin{itemize}
    \item Response surface: $\mu_1(X), \mu_0(X)$
    \item Selection propensity: $P(S=1|X), P(S=0|X)$
    \item Treatment propensity in the observational study: $P(A=1|S=1,X), P(A=0|S=1,X)$
    \item Treatment propensity in the RCT: $P(A=1|S=0), P(A=0|S=0)$
\end{itemize}
The treatment propensity in the RCT is estimated with the empirical probability of treatment within the RCT data.  The response surface, selection propensity and treatment propensity in the observational study are estimated using cross-fitting: the combined data is randomly split into $K = 3$ folds, and the nuisance functions used in each fold are estimated with data out of that fold, using the following models with grid search for hyperparameters. Default hyperparameters in \textit{scikit-learn} for the linear regression model were used. The best hyperparameters found for the gradient boosting classifier, also in \textit{scikit-learn}, were as follows:  \enquote{learning-rate}: 0.01, \enquote{n-estimators}: 50, \enquote{max-depth}: 2, \enquote{min-samples-leaf}: 50, \enquote{min-samples-split}: 50, \enquote{max-features}: \enquote{sqrt} \citep{scikit-learn}. 

\begin{center}
    \begin{tabular}{cccc}
         & Response surface & Selection propensity & Treatment propensity (observational) \\
        \hline
        IDHP & Linear regression & Gradient boosting classifier & Gradient boosting classifier\\
        WHI & Gradient boosting classifier & Gradient boosting classifier & Gradient boosting classifier
    \end{tabular}
\end{center}

As an aside, in Figure \ref{fig:fig2-3-psiwitnessfunc}(a) where we compare the performance of statistics using estimated signals and true signals, we plug in the response surface and selection propensity model implied by our simulation settings into $\psi_i$ to get the true signal difference.

\subsection{Hypothesis testing}
After obtaining the estimated signal difference $\hat{\psi}_i$ by plugging in the estimated nuisance functions into $\psi_i$, the test statistic is calculated as, 
\[n\hat{\fM}_n^2 = \frac{1}{(n-1)}\sum_{i,j \in \cI, i \ne j} \hat{\psi}_i k(x_i, x_j)\hat{\psi}_j\]
where $k(.,.)$ is set as a polynomial kernel of order $3$. One may also use a laplacian kernel or RBF kernel, although we found the polynomial and laplacian kernels to work best in practice. To obtain the $p$-value for the test, we follow \citep{muandet2020kernel} and generate $B=100$ samples of multinomials $\mathbf{w}_k = (w_{k1},w_{k2},\dots,w_{kn})^\top \sim \text{Multinom}(n,(\frac{1}{n},\frac{1}{n},\dots,\frac{1}{n})), k = 1,2,\dots,B$.  For each $k$, we define the bootstrap sample of the null distribution:
\[n\hat{\fM}_{n(k)}^2 = n\sum_{i,j \in \cI, i \ne j} \frac{w_{ki}-1}{n}\hat{\psi}_i k(x_i, x_j)\hat{\psi}_j\frac{w_{kj}-1}{n}\]
The $p$-value is then calculated as
\[\frac{\left[\sum_{k=1}^B \mathbf{1}(n\hat{\fM}_n^2 \le n\hat{\fM}_{n(k)}^2)\right]+1}{B+1}\]
Note that we do not re-estimate the propensity score function in each bootstrap iteration.

\section{Beyond Testing CATE Signals}
In this section, we provide a different formulation of our falsification procedure that tests the potential outcome signals for $\E[Y_a|X], a\in\{0,1\}$, individually, instead of the signal for the contrast $\E[Y_1 - Y_0 | X]$.  This demonstrates that our formulation can be adapted to testing other functions of the potential outcome distribution, other than the one we originally considered.

First, we modify our external validity assumption to accommodate testing individual potential outcomes: 
\begin{restatable}[\textit{External Validity: Observational Study to RCT Transportability of Potential Outcomes}]{assumption}{TransportabilityPO}\label{asmp:transport-alternate}
We assume the following: 
\begin{itemize}[leftmargin=*]
    \item \textit{Mean Exchangeability} — $\E [Y_a | X = x] = \E [Y_a| X = x, S = s]$, $\forall x \in {\cal X}$, $\forall s \in \{0,1\}$, and $\forall a \in \{0,1\}$.
    \item \textit{Positivity of Selection} — $\P(X = x | S = 0) > 0 \implies \P(X = x | S = 1) > 0 $, $\forall x \in {\cal X}$. 
\end{itemize}
\end{restatable}

Now, we will introduce additional notation for our signal functions. Namely, we have the outcome signal from the RCT as follows, 
\begin{align} \label{eq:rct-po-signal}
    \psi_0^a &= \frac{\mathbf{1}\{S = 0\}}{P(S = 0 | X)}\cdot \frac{Y\mathbf{1}\{A = a\}}{P(A = a | S = 0)}, a\in\{0,1\} \nonumber \\
    \boldsymbol\psi_0 &= (\psi_0^0, \psi_0^1)^\top 
\end{align}
Similarly, we have the following outcome signal in the RCT population, but estimated from observational data, as developed in the main paper, 
\begin{align} \label{eq:obs-po-signal}
    \psi_1^a = \frac{1}{P(S = 0 | X)}\bigg[ \mathbf{1}\{S = 0\}\mu_a(X) &+ \mathbf{1}\{S=1\}\frac{P(S=0|X)}{P(S=1|X)}  \frac{\mathbf{1}\{A = a\}(Y - \mu_a(X))}{P(A = a | S = 1, X)}\bigg], a\in\{0,1\} \nonumber      \\
    \boldsymbol\psi_1 &= (\psi_1^0, \psi_1^1)^\top 
\end{align}
Note that the main difference here compared to the main paper is that we define signal functions individually for each potential outcome and then let $\boldsymbol\psi_0$ and $\boldsymbol\psi_1$ be a vector of signals. Now, we have the following proposition, which shows that the vector signals are unbiased for the potential outcomes in the RCT population: 

\begin{restatable}[\textit{Potential Outcome Signals from the RCT and Observational Data}]{proposition}{POSignalRCTOBS}\label{prop:rct-obs-signal}
    Under \cref{asmp:rct-validity} (internal validity of the RCT), the instance-wise potential outcome vector $\boldsymbol\psi_0$ in~\cref{eq:rct-po-signal}, which uses the outcome information from the RCT, is unbiased, i.e., $\E[\boldsymbol\psi_0 | X] = \E[\mathbf{Y} | X, S = 0] = \E[(Y_0,Y_1)^\top | X, S = 0]$. Furthermore, under \cref{asmp:support} and \cref{asmp:transport-alternate}, the instance-wise potential outcome vector $\boldsymbol\psi_1$ in~\cref{eq:obs-po-signal}, which uses the outcome information from the observational data, is unbiased for the potential outcomes in the RCT population, i.e. $\E[\boldsymbol\psi_1 | X] = \E[\mathbf{Y} | X, S = 0] = \E[(Y_0,Y_1)^\top | X, S = 0]$.
\end{restatable}

\begin{proof}
    We first show $\E[\boldsymbol\psi_0|X] = \E[(Y_0,Y_1)^\top|X, S=0]$, i.e. $\E[\psi_0^a|X] = \E[Y_a|X, S=0], a \in \{0, 1\}$
    \begin{align}
        \E[\psi_0^a|X] &= \E\Big[\frac{\mathbf{1}\{S=0\}}{P(S=0|X)}\frac{Y\mathbf{1}\{A=a\}}{P(A=a|S=0)}\Big|X\Big] \nonumber\\
        &= \frac{1}{P(S=0|X)P(A=a|S=0)}\E[\mathbf{1}\{S=0,A=a\}Y|X]\nonumber\\
        &= \frac{1}{P(S=0|X)P(A=a|S=0,X)}\E[\mathbf{1}\{S=0,A=a\}Y|X]\label{eq:pf1-1}\\
        &= \frac{P(S=0,A=a|X)}{P(S=0|X)P(A=a|S=0,X)}\E[Y|X, S=0, A=a]\nonumber\\
        &= \E[Y|X,S=0,A=a] \nonumber\\
        &= \E[Y_a|X,S=0] \label{eq:pf1-2},
    \end{align}
    where (\ref{eq:pf1-1}) is from fixed probability of assignment in Assumption~\ref{asmp:rct-validity}, and (\ref{eq:pf1-2}) is from consistency and ignorability in Assumption~\ref{asmp:rct-validity}.  
    
    We then show $\E[\boldsymbol\psi_1|X] = \E[(Y_0,Y_1)^\top|X, S=1]$, i.e. $\E[\psi_1^a|X] = \E[Y_a|X, S=1], a \in \{0, 1\}$
    \begin{align}
        \E[\psi_1^a|X] &=  \E\Bigg[\frac{1}{P(S = 0 | X)}\bigg[ \mathbf{1}\{S = 0\}\mu_a(X) + \mathbf{1}\{S=1\}\frac{P(S=0|X)}{P(S=1|X)}  \frac{\mathbf{1}\{A = a\}(Y - \mu_a(X))}{P(A = a | S = 1, X)}\bigg]\Bigg|X\Bigg] \nonumber\\
        &= \frac{\E[\mathbf{1}\{S=0\}|X]\mu_a(X)}{P(S=0|X)} + \frac{\E[\mathbf{1}\{S=1, A=a\}(Y-\mu_a(X))|X]}{P(S=1|X)P(A=a|S=1,X)} \nonumber\\
        &=\frac{P(S=0|X)\mu_a(X)}{P(S=0|X)}+\frac{P(S=1,A=a|X)\E[(Y-\mu_a(X))|X,S=1,A=a]}{P(S=1|X)P(A=a|S=1,X)} \nonumber\\
        &= \mu_a(X) + \E[(Y-\mu_a(X))|X,S=1,A=a] \nonumber\\
        &= \mu_a(X) + \E[Y|X,S=1,A=a] - \mu_a(X) \nonumber\\
        &= \mu_a(X) + \E[Y_a|X,S=1] - \mu_a(X) \label{eq:pf1-3}\\
        &= \E[Y_a|X,S=1], \nonumber
    \end{align}
    where (\ref{eq:pf1-3}) is from consistency and ignorability in Assumption~\ref{asmp:support}.  The proposition is now proven.
\end{proof}

We can show a similar corollary to \cref{corr:obs-signal} in the main paper, where we developed the null hypothesis on the CATE signals. Now, we do so for the potential outcome vector signals. Namely, we have, 
\begin{restatable}[\textit{Null Hypothesis on Potential Outcome Difference}]{corollary}{Null Hypothesis on Potential Outcome Difference}\label{corr:obs-signal-alternate} Define $\boldsymbol\psi = \boldsymbol\psi_1 - \boldsymbol\psi_0$ as the instance-wise signal difference between the observational and RCT potential outcome estimates. Then, under the null hypothesis, i.e. under \cref{asmp:rct-validity,asmp:support} and \cref{asmp:transport-alternate}, we have it that $\E[\boldsymbol\psi|X] = \mathbf{0}$.
\end{restatable}
\begin{proof}
If~\cref{asmp:rct-validity,asmp:support} and \cref{asmp:transport-alternate} hold, then~\cref{prop:rct-obs-signal} implies that ${\E[\boldsymbol\psi_0 | X] = \E[\boldsymbol\psi_1 | X]} = \E[\mathbf{Y} | X, S=0] = \E[(Y_0,Y_1)^\top | X, S=0]$. 
\end{proof}

Our assumptions, i.e. \cref{asmp:rct-validity}, \cref{asmp:support}, and \cref{asmp:transport-alternate}, give us a set of conditional moment restrictions (CMRs) on the signal difference, $\boldsymbol\psi$, which is a difference of vector signals: 
\begin{align} \label{eq:null-po}
    H_0: \E[\boldsymbol\psi | X ] = \boldsymbol{0}\quad P_X\textit{-almost surely}
\end{align}
As before, $P_X$ is the distribution of $X$ on the joint distribution of the RCT and observational study. By the law of iterated expectations, akin to the development in \cref{prop:moment}, \cref{eq:null-po} implies an infinite set of unconditional moment restrictions,
\begin{equation}\label{eq:po-umr}
    \E[\boldsymbol\psi^\top \boldsymbol{f}(X)] = 0, \forall \boldsymbol{f} \in \mathcal{F} \times \mathcal{F},
\end{equation}
 where $\mathcal{F}$ is the set of measurable functions on $\mathcal{X}$. Note that now, $\boldsymbol{f}$ is a vector-valued function, where $\boldsymbol{f}(X) = (f_0(X), f_1(X))^\top$. Now, as in the main paper, we follow the CMR testing procedure presented in \citet{muandet2020kernel}, where we let $\mathcal{F}$ be a RKHS and use the maximum moment restriction (MMR) within the unit ball of the RKHS as our test statistic. Following this, we present the following theorem, which is a modified version of \cref{thm:mmr}. 

\begin{restatable}[\textit{Maximum Moment Restriction-based test for Potential Outcomes}]{theorem}{thmmmr-po}\label{thm:mmr-po}
     Let $\cF$ be a RKHS with reproducing kernel $k(\cdot,\cdot): \cX \times \cX\rightarrow \R$ that is ISPD, continuous and bounded, equipped with inner product $\langle. , .\rangle_\cF$.  Denote $\cF^2$ as the product RKHS $\cF\times\cF$ equipped with inner product $\langle \boldsymbol{f} , \boldsymbol{g}\rangle_{\cF^2} = \langle (f_1, f_2)^\top , (g_1, g_2)^\top\rangle_{\cF^2} := \langle f_1, g_1 \rangle_\cF + \langle f_2, g_2 \rangle_\cF$. Suppose the elements of $|\E[\boldsymbol\psi|X]| < \infty$ almost surely in $P_X$, and $\E[[k(X, X') \boldsymbol\psi^\top  \boldsymbol\psi']^2] < \infty$ where $(\boldsymbol\psi', X')$ is an independent copy of $(\boldsymbol\psi, X)$. Let $\fM^2 = \sup_{\boldsymbol{f} \in \cF^2, ||\boldsymbol{f}|| \le 1}(\E[\boldsymbol\psi^\top \boldsymbol{f}(X)])^2$. Then,
     \begin{enumerate}
         \item The conditional moment testing problem in Eq. \ref{eq:null-po} can be reformulated in terms of the MMR as $H_0^': \fM^2 = 0$, $H_1^': \fM^2 \ne 0$.
     \end{enumerate}
     Further, let the test statistic be the empirical estimate of $\fM^2$,
     \[\hat{\fM}_n^2 = \frac{1}{n(n-1)}\sum_{i,j \in \cI, i \ne j} k(x_i, x_j)\boldsymbol\psi_i^\top \boldsymbol\psi_j\]
     \begin{enumerate}
         \setcounter{enumi}{1}
         \item Then, under $H_0^'$, %
         \[ n\hat{\fM}_n^2 \xrightarrow[]{d} \sum_{j=1}^\infty \lambda_j(Z_j^2 - 1)\]
         where $Z_j$ are independent standard normal variables and $\lambda_j$ are the eigenvalues for $k(x,x') \boldsymbol\psi^\top \boldsymbol\psi'$. 
         \item Under $H_1^'$,
         \[ \sqrt{n}(\hat{\fM}_n^2 - \fM^2) \xrightarrow[]{d} \cN(0, 4\sigma^2)\]
         where $\sigma^2= var_{(\boldsymbol\psi, X)}[\E_{(\boldsymbol\psi', X')}[k(X,X') \boldsymbol\psi^\top \boldsymbol\psi']]$
     \end{enumerate}
\end{restatable}

\begin{proof}
    The proof is very similar to how we proved Therorem~\ref{thm:mmr}.  Let us define the following operator,
    \begin{equation} \label{eq:m-operator-po}
        M \boldsymbol{f} = \E [\boldsymbol{\psi}^\top \boldsymbol{f}(X)]
    \end{equation}
    where $\boldsymbol{f} \in \cF^2$.  Since the elements of $|\E[\boldsymbol{\psi}|X]| < \infty$ almost surely in $P_X$, $M$ is a bounded linear operator. By Riesz representation theorem, there exists a unique $\boldsymbol{g} \in \cF^2$ such that
    \begin{equation*}
       M \boldsymbol{f} = \langle \boldsymbol{f},\boldsymbol{g} \rangle_{\cF^2}
    \end{equation*}
    where 
    \begin{equation*}
    \boldsymbol{g} = \E[\boldsymbol{\psi} k(X,\cdot)].
    \end{equation*}
    Therefore, it follows that
    \begin{equation*}
        \fM^2 = \sup_{\boldsymbol{f} \in \cF^2, \|\boldsymbol{f}\| \leq 1} \left(\E[\boldsymbol{\psi}^\top \boldsymbol{f}(X)]\right)^2 = \sup_{\boldsymbol{f} \in \cF^2, \|\boldsymbol{f}\| \leq 1} \langle \boldsymbol{f}, \boldsymbol{g} \rangle_{\cF^2}^2 = \left\langle \frac{\boldsymbol{g}}{\|\boldsymbol{g}\|}, \boldsymbol{g} \right\rangle_{\cF^2}^2 = \|\boldsymbol{g}\|^2
    \end{equation*}

    Since $\fM^2 = \|\boldsymbol{g}\|^2$, the first statement in Theorem~\ref{thm:mmr-po} is essentially \[\E[\boldsymbol{\psi} | X] = \boldsymbol{0}, P_X\text{-almost surely} \Leftrightarrow \|\boldsymbol{g}\|^2 = 0\]
    That is, $\boldsymbol{g} \in \cF^2$ fully captures the information of the CMR for all $x \in \cX$.  This equivalence, which we will now prove, is crucial since our statistical test is based on $\|\boldsymbol{g}\|^2$ and its estimates, while Corollary~\ref{corr:obs-signal-alternate} is directed to the CMR:
    
    ($\Rightarrow$) We note that since $\cF^2$ is a Hilbert space, it follows that $\boldsymbol{g} \in \cF^2$, and from (\ref{eq:po-umr}), $\forall \boldsymbol{f} \in \cF^2, \langle \boldsymbol{f},\boldsymbol{g} \rangle_{\cF^2} = \E[\boldsymbol{\psi}^\top \boldsymbol{f}(X)] = 0$. $\boldsymbol{g}$ can now only be a zero vector.  Therefore, $\|\boldsymbol{g}\|^2 = 0$.

    ($\Leftarrow$) 
    \begin{align*}
        &\|\boldsymbol{g}\|^2 = 0\\
        \Rightarrow~ & \left\|\E[\boldsymbol{\psi} k(X,.)]\right\|^2 = 0 \\
        \Rightarrow~ & \left\|\E[\E[\boldsymbol{\psi}|X]k(X,.)]\right\|^2 = 0\\
        \Rightarrow~ & \left\|\int_{\cX}k(x,.)\E[\boldsymbol{\psi}|x] p_{X}(x) dx\right\|^2 = 0\\
        \Rightarrow~ & \iint_{\cX \times \cX} p_{X}(x)\E[\boldsymbol{\psi}^\top|x]k(x,x^')\E[\boldsymbol{\psi}|x']p_{X}(x^') dx dx' = 0\\
        \Rightarrow~ & \|\E[\boldsymbol{\psi}|x]p_{X}(x)\|^2 = 0 \qquad (\because k(\cdot,\cdot)\text{ is ISPD})\\
        \Rightarrow~ & \E[\boldsymbol{\psi}|x] = \boldsymbol{0}, P_{X}\text{-almost surely}\\
    \end{align*}
    
    Finally, we move to the second and third statements of Theorem~\ref{thm:mmr-po}, which define the estimator and its statistical properties.  Since $\fM^2 = \|\boldsymbol{g}\|^2 = \|\E[\boldsymbol{\psi} k(X,. )]\|^2 = \E[\E[\boldsymbol{\psi}^\top k(X,X') \boldsymbol{\psi}']]$ where $(X, \boldsymbol{\psi})$ and $(X', \boldsymbol{\psi}')$ are independently and identically distributed, we may use a $U$-statistic to estimate $\fM^2$, which is exactly
    \[\hat{\fM}_n^2 = \frac{1}{n(n-1)}\sum_{i,j \in \cI, i \ne j} \boldsymbol{\psi}_i^\top k(x_i, x_j)\boldsymbol{\psi}_j\]
    
    Now note that in Lemma~\ref{lem:Serfling}, if we set $W = (\boldsymbol{\psi}, X)$, $h(W,W') = \boldsymbol{\psi}^\top k(X,X') \boldsymbol{\psi}'$, $\theta = \fM^2$, $\zeta_1 = \sigma^2$, the second and third statements of Theorem~\ref{thm:mmr} holds as long as $\fM^2 = 0 \Leftrightarrow \sigma^2 = var_{(\boldsymbol{\psi}, X)}[\E_{(\boldsymbol{\psi}', X')}[\boldsymbol{\psi}^\top k(X,X')\boldsymbol{\psi}']] = 0$, which we will now show:
    
    ($\Rightarrow$)
    \[\E_{(\boldsymbol{\psi}', X')}[\boldsymbol{\psi}^\top k(X,X') \boldsymbol{\psi}'] = \langle \boldsymbol{\psi} k(X,.) , \E_{(\boldsymbol{\psi}', X')}[\boldsymbol{\psi}' k(X',.)]\rangle_{\cF^2} = \|\boldsymbol{\psi} k(X,.)\| \left\langle \frac{\boldsymbol{\psi} k(X,.)}{\|\boldsymbol{\psi} k(X,.)\|} , \boldsymbol{g}\right\rangle_{\cF^2}\]
    Now since 
    \[\frac{\boldsymbol{\psi} k(X,.)}{\|\boldsymbol{\psi} k(X,.)\|} \in \cF^2, \left\|\frac{\boldsymbol{\psi} k(X,.)}{\|\boldsymbol{\psi} k(X,.)\|}\right\| = 1\] 
    and 
    \[\fM^2 = 0 \Rightarrow \sup_{\boldsymbol{f}\in\cF^2, \|\boldsymbol{f}\| \le 1}\langle \boldsymbol{f}, \boldsymbol{g}\rangle_{\cF^2} = 0 \Rightarrow \langle \boldsymbol{f}, \boldsymbol{g}\rangle_{\cF^2} = 0, \forall \boldsymbol{f}\in\cF^2, \|\boldsymbol{f}\| \le 1\]
    We conclude 
    \[\fM^2 = 0 \Rightarrow \left\langle \frac{\boldsymbol{\psi} k(X,.)}{\|\boldsymbol{\psi} k(X,.)\|} , \boldsymbol{g} \right\rangle_{\cF^2} = 0 \Rightarrow \E_{(\boldsymbol{\psi}', X')}[\boldsymbol{\psi}^\top k(X,X')\boldsymbol{\psi}'] = 0 \Rightarrow var_{(\boldsymbol{\psi}, X)}[\E_{(\boldsymbol{\psi}', X')}[\boldsymbol{\psi}^\top k(X,X')\boldsymbol{\psi}']] = 0\]
    
    ($\Leftarrow$)

    We first note that $var_{(\boldsymbol{\psi}, X)}(\E_{(\boldsymbol{\psi}', X')}[\boldsymbol{\psi}^\top k(X,X') \boldsymbol{\psi}']) = 0$ implies that $\E_{(\boldsymbol{\psi}', X')}[\boldsymbol{\psi}^\top k(X,X') \boldsymbol{\psi}']$ is a constant $P_{(\boldsymbol{\psi}, X)}$-almost surely.  We denote this constant as $c$ so we have
    \begin{equation}
        \E_{(\boldsymbol{\psi}', X')}[\boldsymbol{\psi}^\top k(X,X')\boldsymbol{\psi}'] = c, P_{(\boldsymbol{\psi},X)}\text{-almost surely}
        \label{eq:MMR-po-leftarrow}
    \end{equation}

    From the definition of $\boldsymbol{\psi}$, let $X = x^*$ be in the support of the observational study, then
    \begin{align*}
        \E[\boldsymbol{\psi}| S=1, X=x^*] &= \frac{1}{P(S=1|X=x^*)}\E\bigg[\bigg(\frac{\mathbf{1}(A=1)(Y-\mu_1(x^*))}{P(A=1|S=1,X=x^*)}, \\
        &\qquad\qquad \frac{\mathbf{1}(A=0)(Y-\mu_0(x^*))}{P(A=0|S=1,X=x^*)}\bigg)^\top \bigg|S=1,X=x^*\bigg]\\
        &=\frac{1}{P(S=1|X=x^*)}(\E[Y-\mu_1(x^*)|A=1, S=1, X=x^*], \\
        &\qquad\qquad\E[Y-\mu_0(x^*)|A=0, S=1, X=x^*)])^\top\\
        &= \boldsymbol{0},
    \end{align*}
    where the last equality stems from the definition of $\mu_1$ and $\mu_0$.  Now note that 
    \begin{align*}
      \E_{\boldsymbol{\psi}}[\E_{(\boldsymbol{\psi}', X')}[\boldsymbol{\psi}^\top k(X, X')\boldsymbol{\psi}'|S=1, X = x^*]] &= \E_{\boldsymbol{\psi}}[\E_{(\boldsymbol{\psi}', X')}[\boldsymbol{\psi}^\top k(x^*, X') \boldsymbol{\psi}'|S=1, X = x^*]]\\
      &= (\E_{\boldsymbol{\psi}}[\boldsymbol{\psi}| S=1, X=x^*])^\top\E_{(\boldsymbol{\psi}', X')}[k(x^*, X') \boldsymbol{\psi}'] \\
      &= \boldsymbol{0}^\top \E_{(\boldsymbol{\psi}', X')}[k(x^*, X') \boldsymbol{\psi}'] = 0
    \end{align*}
    But also we have, from (\ref{eq:MMR-po-leftarrow}),
    \[\E_{\boldsymbol{\psi}}[\E_{(\boldsymbol{\psi}', X')}[\boldsymbol{\psi}^\top k(X, X')\boldsymbol{\psi}'|S=1, X = x^*]] = \E_{\boldsymbol{\psi}}[c] = c\]
    Therefore, we have $c = 0$ and thus
    \[\fM^2 = \E_{(\boldsymbol{\psi}, X)}[\E_{(\boldsymbol{\psi}', X')}[\boldsymbol{\psi}^\top k(X, X')\boldsymbol{\psi}']] = \E_{(\boldsymbol{\psi}, X)}[0] = 0\]
    so this side of the arrow is also proven.

\end{proof}
We label this alternate formulation, where we test on the potential outcomes directly instead of the contrast, as \textbf{MMR-Absolute}. We give the rejection rate of \textbf{MMR-Absolute} under different amounts of selection bias induced in the WHI dataset in \cref{tab:whi_table_app}. We find that the \textbf{MMR-Absolute} approach vastly over-rejects, indicating the utility of testing the causal contrast as opposed to the absolute potential outcomes. 
\begin{table}[h]
\centering
{\footnotesize
\begin{tabular}{lcccc}
\toprule
\toprule
    \textit{Selection Bias}          & \textbf{MMR-Contrast}     & \textbf{MMR-Absolute}    &\textbf{ATE}     & \textbf{GATE}  \\ \midrule
\quad  $p=0$  & 0.29 & 1.0 & 0.32 & \textbf{0.17}  \\
\quad  $p=0.05$  & \textbf{0.67} & 1.0 & 0.58 & 0.40 \\
\quad  $p=0.10$  & \textbf{0.94} & 1.0 & 0.88 & 0.67 \\
\quad  $p=0.15$  & \textbf{1.0} & 1.0 & 0.98  & 0.91 \\
\bottomrule
\bottomrule
\end{tabular}}
\caption{Rejection rate when introducing different amounts of selection bias into the observational data in WHI study. $p$ stands for the strength of selection introduced in the the data (refer to Section~\ref{sec:whi} for details).}
 \label{tab:whi_table_app}
\end{table}

\section{When does testing for bias across subgroups improve power?}
In our experimental results, we find that the GATE approach has limited power compared to the ATE approach. Indeed, the performance of GATE versus ATE depends in part on the choice of subgroups used for GATE. In the extreme, if the difference in effect is identical across all subgroups, testing for differences in ATE may have higher power once multiple-testing corrections are applied. 
To build intuition, we will provide a simple example for when a GATE-based test might have higher power compared to an ATE-based test. We will then formalize this example and provably show under what conditions a GATE-based test would have higher asymptotic power compared to an ATE-based test. When referring to the test that tests differences of GATEs or ATEs, we will use the bold form: \textbf{GATE} and \textbf{ATE}. When referring to the causal quantity itself, we will simply use GATE and ATE.

\begin{figure*}[t!]
\centering
\includegraphics[width=0.9\textwidth]{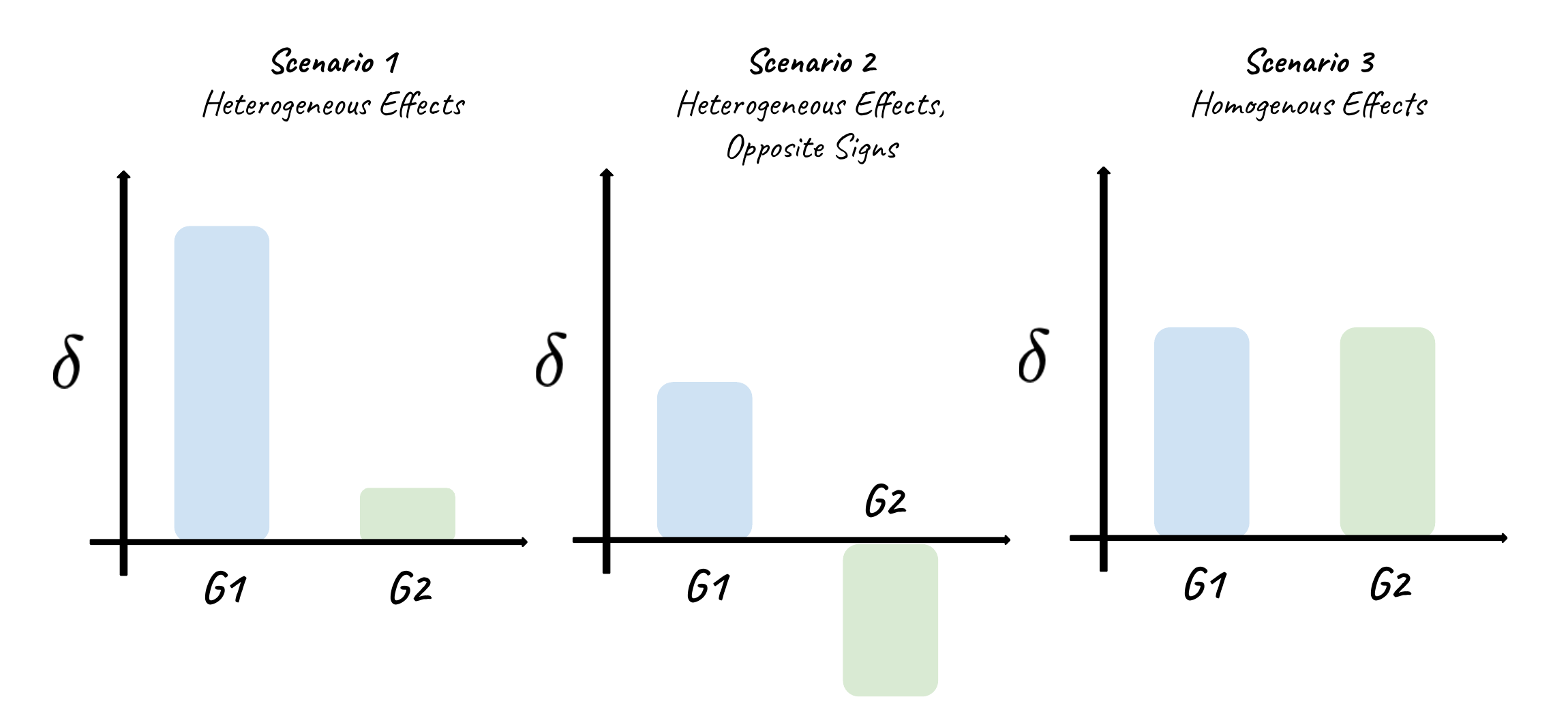}
\caption{Barplot depiction of three toy scenarios, where we plot the asymptotic bias, denoted by $\delta$ (see \Cref{eq:delta1,eq:delta2}), of the observational estimator in each subgroup. Our goal is to detect, from finite samples, whether or not this asymptotic bias is non-zero for any subgroup. In scenario 3, pooling the data and testing for the overall bias (the \textbf{ATE} approach) yields better power than testing for differences across subgroups.  Explicitly testing the bias in each subgroup (the \textbf{GATE} approach) is beneficial in scenarios like 1 and 2 where heterogeneity exists. The x-axis contains the group name, and the y-axis indicates the magnitude of $\delta$.}
\label{fig:gate-power}
\end{figure*}

\textbf{Toy Example}: To build intuition, we will use a toy example to construct three scenarios in which the asymptotic power between \textbf{GATE} and \textbf{ATE} may differ. Consider testing whether there is bias in a population, where the null hypothesis is that the population mean is zero. Let there be two subgroups in the population, \textbf{G1} and \textbf{G2}. Finally, let $\delta$ be a term denoting the asymptotic bias. Figure \ref{fig:gate-power} shows three separate scenarios: 
\begin{itemize}
    \item In scenario 1, the bias in \textbf{G1} is significantly higher than the bias in \textbf{G2}. As we formalize below, \textbf{GATE} will have higher power than \textbf{ATE} as $|\delta|$ gets larger and the sample size of \textbf{G1} is reasonable. See below for precise conditions. 
    \item In scenario 2, the bias in the two subgroups have the same magnitude but are in opposite directions. Below, we show that \textbf{GATE} has better power than \textbf{ATE} in this scenario, given a large enough $|\delta|$ to overcome the penalty of multiple hypothesis testing. This result is intuitive since testing differences in ATE would fail to reject the null since the average effect over the entire population would be close to zero.
    \item In scenario 3, the bias is the same magnitude and direction in both subgroups. We show below that the \textbf{ATE} has better power than \textbf{GATE} regardless of what the magnitude of $\delta$ is. Intuitively, pooling together the two subgroups would yield a larger sample to detect the bias.
\end{itemize} 
In the subsequent paragraphs, we will formalize these three scenarios in the context of our setting, where we have estimates from observational and RCT data. Note that the theoretical framework that we introduce below covers these three scenarios as well as others. 

\subsection{Notation and Assumptions}
We recall some notation and definitions from \citep{hussain2022falsification}. 
\begin{definition}[GATE, \citet{hussain2022falsification}]\label{def:group_effect}
We define the group average treatment effect (GATE) as
\begin{equation}\label{eq:definition_gate}
\tau_i \coloneqq \E[Y_1 - Y_0 \mid G = i, S = 0]
\end{equation}
where $G$ is the group indicator variable taking values $\{1, 2\}$, and $S = 0$ indicates the RCT population. 
\end{definition}
The GATE estimator for subgroup $i$ using RCT data will be denoted, $\hat{\tau}_i(0)$, while the estimator using observational data will be denoted, $\hat{\tau}_i(1)$.
\begin{definition}[ATE]\label{def:ate}
We define the average treatment effect (ATE) as
\begin{equation}\label{eq:definition_ate}
\tau \coloneqq \E[Y_1 - Y_0 \mid S = 0]
\end{equation}
where $S = 0$ indicates the RCT population. 
\end{definition}
Akin to the GATE estimators, the ATE estimator using RCT data will be denoted, $\hat{\tau}(0)$, while the estimator using observational data will be denoted, $\hat{\tau}(1)$.
Writing $\rho_{i0}$ ($\rho_{i1}$) as the proportion of observations in the RCT (the obserational study) that belongs to subgroup $i$, we then modify Assumption 2.4 from \citet{hussain2022falsification} as follows, : 
\begin{assumption}\label{asmp:all_asymptotic_normal}
  All GATE estimators are pointwise asymptotically normally distributed and independent
  \begin{align}
    \sqrt{\rho_{i0} N_0} (\hat{\tau}_i(0) - \tau_i(0))/\hat{\sigma}_i(0) \cid \cN(0, 1)\\
    \sqrt{\rho_{i1} N_1} (\hat{\tau}_i(1) - \tau_i(1))/\hat{\sigma}_i(1) \cid \cN(0, 1) 
  \end{align}
  Here, $\cid$ denotes convergence in distribution, and $\hat{\sigma}^2_i(k)$ is an estimate of the variance that converges in probability to $\sigma^2_i(k)$, the asymptotic variance of $\sqrt{\rho_{ik} N_k}(\tauhki - \tauki)$, for $k=0$ and $k=1$.
\end{assumption}

In addition to assumptions on the GATE estimators, we also have assumptions on the asymptotic distributions of the ATE estimators for both studies:

\begin{assumption}\label{asmp:ate_asymptotic_normal}
    Both ATE estimators are asymptotically normally distributed and independent
    \begin{align}
        \sqrt{N_0} (\hat{\tau}(0) - \tau(0))/\hat{\sigma}(0) \cid \cN(0, 1)\\
        \sqrt{N_1} (\hat{\tau}(1) - \tau(1))/\hat{\sigma}(1) \cid \cN(0, 1)
    \end{align}
    where $\hat{\sigma}^2(k)$ is an estimate of the variance that converges in probability to $\sigma^2(k)$, the asymptotic variance of $\sqrt{N_k}(\hat{\tau}(k) - \tau(k))$, for $k=0$ and $k=1$.
\end{assumption}

\subsection{Theoretical Example}
Given the assumptions and definitions, we present a formal example:
\begin{example}\label{example:gate-ate-power}
Suppose there are two subgroups in the RCT and observational study.  To reflect the consistency of the RCT GATE estimators and quantify the bias of the GATE estimators from the observational study, we define 
\begin{align}
    (RCT, \text{group 1}) & \quad \tau_1(0) = \tau_1\\
    (RCT, \text{group 2}) & \quad \tau_2(0) = \tau_2\\
    (OBS, \text{group 1}) & \quad \tau_1(1) = \tau_1 + \delta_1\label{eq:delta1}\\ 
    (OBS, \text{group 2}) & \quad \tau_2(1) = \tau_2 + \delta_2 \label{eq:delta2}
\end{align}
For simplicity, we assume that in both the RCT and observational study, half of the population is in group 1 and half of the population is in group 2, i.e. $\rho_{i0} = \rho_{i1} = 1/2$ for $i = 1, 2$.  Then we have
\begin{align}
    \tau(0) &= \frac{\tau_1(0)+\tau_2(0)}{2} = \frac{\tau_1+\tau_2}{2} \label{eq:tau_0}\\
    \tau(1) &= \frac{\tau_1(1)+\tau_2(1)}{2} = \frac{\tau_1+\delta_1+\tau_2+\delta_2}{2} \label{eq:tau_1}
\end{align}
Lastly, we introduce the following shorthand notations, writing the total sample size $N = N_0 + N_1$ and letting $N_0 = \rho N$, $N_1 = (1-\rho) N$:
\begin{align}
    \boldsymbol{\sigma} &= \sqrt{N_0 + N_1} \sqrt{\frac{\sigma^2(0)}{N_0}+\frac{\sigma^2(1)}{N_1}} = \sqrt{\frac{\sigma^2(0)}{\rho}+\frac{\sigma^2(1)}{1-\rho}} \label{eqn:def_sigma}\\
    \boldsymbol{\sigma_1} &= \sqrt{\rho_{10}N_0 + \rho_{11}N_1} \sqrt{\frac{\sigma_1^2(0)}{\rho_{10}N_0}+\frac{\sigma_1^2(1)}{\rho_{11}N_1}} = \sqrt{\frac{\sigma_1^2(0)}{\rho}+\frac{\sigma_1^2(1)}{1-\rho}} \label{eqn:def_sigma_1}\\
    \boldsymbol{\sigma_2} &= \sqrt{\rho_{20}N_0 + \rho_{21}N_1} \sqrt{\frac{\sigma_2^2(0)}{\rho_{20}N_0}+\frac{\sigma_2^2(1)}{\rho_{21}N_1}} = \sqrt{\frac{\sigma_2^2(0)}{\rho}+\frac{\sigma_2^2(1)}{1-\rho}} \label{eqn:def_sigma_2}
\end{align}
\end{example}

To simplify the development, we will make the following assumption for this example: 
\begin{assumption}\label{asmp:variances_equal}
    Assume that $\sigma^2(0) = \sigma_1^2(0) = \sigma_2^2(0)$ and $\sigma^2(1) = \sigma_1^2(1) = \sigma_2^2(1)$, so that we can write \Cref{eqn:def_sigma,eqn:def_sigma_1,eqn:def_sigma_2} as, 
    \begin{equation}
        \boldsymbol{\sigma} = \boldsymbol{\sigma_1} = \boldsymbol{\sigma_2} = \sqrt{\frac{\sigma^2(0)}{\rho}+\frac{\sigma^2(1)}{1-\rho}} \label{eqn:def_sigma_simp}
    \end{equation}
\end{assumption}

The asymptotic power of the \textbf{ATE} and \textbf{GATE} can then be given by the following propositions:
\begin{restatable}[\textit{Asymptotic power of \textbf{ATE}}]{proposition}{ATEPOWER}
    \label{prop:ate-power}
    Under \Cref{asmp:variances_equal}, the asymptotic power of \textbf{ATE} as $N \rightarrow \infty$ (holding $\rho$ as constant) is given by
    \[1-\left[\Phi\Bigg(\frac{|\frac{\delta_1+\delta_2}{2}|}{\boldsymbol{\sigma}/\sqrt{N}}+z_{\alpha/2}\Bigg) - \Phi\Bigg(\frac{|\frac{\delta_1+\delta_2}{2}|}{\boldsymbol{\sigma}/\sqrt{N}}-z_{\alpha/2}\Bigg)\right]\]
\end{restatable}

\begin{proof}
    From Proposition 2.1 of \citep{hussain2022falsification}, given the asymptotic distributions from Assumption~\ref{asmp:ate_asymptotic_normal} and Equations~(\ref{eq:tau_0}), (\ref{eq:tau_1}), we have $\tau(1)-\tau(0) = \frac{\delta_1+\delta_2}{2}$ and thus
    \begin{equation}
        \frac{\hat{\tau}(1)-\hat{\tau}(0)-\frac{\delta_1+\delta_2}{2}}{\hat{\sigma}/\sqrt{N}} \cid \cN(0,1) \label{eq:asymp_ate}
    \end{equation}
    which allows us to construct a $Z$-test on the null hypothesis $H_0: \frac{\delta_1+\delta_2}{2} = 0$ based on the rejection region
    \[\Bigg|\frac{\hat{\tau}(1)-\hat{\tau}(0)}{\hat{\sigma}/\sqrt{N}}\Bigg|>z_{\alpha/2}\]
    The asymptotic power of the $Z$-test under the alternative hypothesis distribution shown in (\ref{eq:asymp_ate}) is then, from Theorems 10.4, 10.6 in \citep{Wasserman2004-qm}
    \[1-\Phi\Bigg(\frac{|\frac{\delta_1+\delta_2}{2}|}{\boldsymbol{\sigma}/\sqrt{N}}+z_{\alpha/2}\Bigg) + \Phi\Bigg(\frac{|\frac{\delta_1+\delta_2}{2}|}{\boldsymbol{\sigma}/\sqrt{N}}-z_{\alpha/2}\Bigg) = 1-\left[\Phi\Bigg(\frac{|\frac{\delta_1+\delta_2}{2}|}{\boldsymbol{\sigma}/\sqrt{N}}+z_{\alpha/2}\Bigg) - \Phi\Bigg(\frac{|\frac{\delta_1+\delta_2}{2}|}{\boldsymbol{\sigma}/\sqrt{N}}-z_{\alpha/2}\Bigg)\right]\]
\end{proof}

\begin{restatable}[\textit{Asymptotic power of \textbf{GATE}}]{proposition}{GATEPOWER}
    \label{prop:gate-power}
    Under \Cref{asmp:variances_equal}, the asymptotic power of \textbf{GATE} is given by
    \[1-\left[\Phi\Bigg(\frac{1}{\sqrt{2}}\frac{|\delta_1|}{\boldsymbol{\sigma}/\sqrt{N}}+z_{\alpha/4}\Bigg) - \Phi\Bigg(\frac{1}{\sqrt{2}}\frac{|\delta_1|}{\boldsymbol{\sigma}/\sqrt{N}}-z_{\alpha/4}\Bigg)\right]\left[\Phi\Bigg(\frac{1}{\sqrt{2}}\frac{|\delta_2|}{\boldsymbol{\sigma}/\sqrt{N}}+z_{\alpha/4}\Bigg) - \Phi\Bigg(\frac{1}{\sqrt{2}}\frac{|\delta_2|}{\boldsymbol{\sigma}/\sqrt{N}}-z_{\alpha/4}\Bigg)\right]\]
\end{restatable}

\begin{proof}
    With arguments similar to Proposition~\ref{prop:ate-power}, since the total sample size for subgroup $i$ is $\rho_{i0}N_0 + \rho_{i1}N_1 = N/2$, the asymptotic power of the $Z$-test comparing the GATE estimates for group $i$ would be ($i \in \{1,2\}$)
    \begin{align*}
        \xi_i &= 1-\left[\Phi\Bigg(\frac{|\delta_i|}{\boldsymbol{\sigma_i}/\sqrt{N/2}}+z_{\alpha/4}\Bigg) - \Phi\Bigg(\frac{|\delta_i|}{\boldsymbol{\sigma_i}/\sqrt{N/2}}-z_{\alpha/4}\Bigg)\right]\\
        &= 1-\left[\Phi\Bigg(\frac{1}{\sqrt{2}}\frac{|\delta_i|}{\boldsymbol{\sigma}/\sqrt{N}}+z_{\alpha/4}\Bigg) - \Phi\Bigg(\frac{1}{\sqrt{2}}\frac{|\delta_i|}{\boldsymbol{\sigma}/\sqrt{N}}-z_{\alpha/4}\Bigg)\right]
    \end{align*}
        where the last equality stems from \Cref{asmp:variances_equal}. Since we are rejecting the null hypothesis of $H_0: \delta_1 = 0 \text{ and } \delta_0 = 0$ when the test in either subgroup shows significance, and the two tests are independent, the power of \textbf{GATE} is then 
    \begin{align*}
        &1-(1-\xi_1)(1-\xi_2)\\
        = &1-\left[\Phi\Bigg(\frac{1}{\sqrt{2}}\frac{|\delta_1|}{\boldsymbol{\sigma}/\sqrt{N}}+z_{\alpha/4}\Bigg) - \Phi\Bigg(\frac{1}{\sqrt{2}}\frac{|\delta_1|}{\boldsymbol{\sigma}/\sqrt{N}}-z_{\alpha/4}\Bigg)\right]\left[\Phi\Bigg(\frac{1}{\sqrt{2}}\frac{|\delta_2|}{\boldsymbol{\sigma}/\sqrt{N}}+z_{\alpha/4}\Bigg) - \Phi\Bigg(\frac{1}{\sqrt{2}}\frac{|\delta_2|}{\boldsymbol{\sigma}/\sqrt{N}}-z_{\alpha/4}\Bigg)\right]
    \end{align*}
\end{proof}

We now investigate three scenarios regarding the pattern of bias for the GATE estimators from the observational study:
\newline
\newline
\textbf{Scenario 1}: Only the GATE estimator for subgroup 1 is biased
\newline

This scenario can be depicted by letting $\delta_1 = \delta \ne 0$ and $\delta_2 = 0$, so that we have $\frac{\delta_1+\delta_2}{2} = \frac{\delta}{2}$.  The power of \textbf{ATE} and \textbf{GATE} in this scenario can be given by, based on Propositions~\ref{prop:ate-power} and~\ref{prop:gate-power}: 
\begin{align}
    \xi_{\text{\textbf{ATE}}} &= 1-\left[\Phi\Bigg(\frac{|\delta/2|}{\boldsymbol{\sigma}/\sqrt{N}}+z_{\alpha/2}\Bigg) - \Phi\Bigg(\frac{|\delta/2|}{\boldsymbol{\sigma}/\sqrt{N}}-z_{\alpha/2}\Bigg)\right]\\
    \xi_{\text{\textbf{GATE}}} &= 1 - [\Phi(z_{\alpha/4}) - \Phi(-z_{\alpha/4})]\left[\Phi\Bigg(\frac{1}{\sqrt{2}}\frac{|\delta|}{\boldsymbol{\sigma}/
    \sqrt{N}}+z_{\alpha/4}\Bigg) - \Phi\Bigg(\frac{1}{\sqrt{2}}\frac{|\delta|}{\boldsymbol{\sigma}/
    \sqrt{N}}-z_{\alpha/4}\Bigg)\right] \nonumber\\
    & = 1-\Big(1-\frac{\alpha}{2}\Big)\left[\Phi\Bigg(\frac{1}{\sqrt{2}}\frac{|\delta|}{\boldsymbol{\sigma}/
    \sqrt{N}}+z_{\alpha/4}\Bigg) - \Phi\Bigg(\frac{1}{\sqrt{2}}\frac{|\delta|}{\boldsymbol{\sigma}/
    \sqrt{N}}-z_{\alpha/4}\Bigg)\right] \nonumber\\
    & = 1-\Big(1-\frac{\alpha}{2}\Big)\left[\Phi\Bigg(\sqrt{2}\frac{|\delta/2|}{\boldsymbol{\sigma}/
    \sqrt{N}}+z_{\alpha/4}\Bigg) - \Phi\Bigg(\sqrt{2}\frac{|\delta/2|}{\boldsymbol{\sigma}/
    \sqrt{N}}-z_{\alpha/4}\Bigg)\right]
\end{align}

Denoting $\delta^* := \frac{|\delta/2|}{\boldsymbol{\sigma}/\sqrt{N}} \ge 0$, we may simplify the expressions as,
\begin{align}
    \xi_{\text{\textbf{ATE}}} &= 1- \left[\Phi(\delta^* + z_{\alpha/2})-\Phi(\delta^* - z_{\alpha/2})\right]\\
    \xi_{\text{\textbf{GATE}}} &= 1 - \Big(1-\frac{\alpha}{2}\Big) \bigg[\Phi\big(\sqrt{2}\delta^* + z_{\alpha/4}\big)-\Phi\big(\sqrt{2}\delta^* - z_{\alpha/4}\big)\bigg]
\end{align}

Before we derive sufficient conditions for $\xi_{\text{\textbf{GATE}}} > \xi_{\text{\textbf{ATE}}} $, we state the following lemma on the properties of $\Phi(.)$ and $\Phi^{-1}(.)$:

\begin{lemma}
    $\forall \alpha \in (0,1), \frac{z_{\alpha/4}}{z_{\alpha/2}} <  \frac{z_{1/4}}{z_{1/2}} \approx 1.1185$.
    \label{lem:z-ratio}
\end{lemma}
\begin{lemma}
    $\forall a > 1$, $\Phi(ax)-\Phi(x)$ is a strictly decreasing function in $x$ as $x > \sqrt{\frac{2\log a}{a^2-1}}$.
    \label{lem:psi-ineq}
\end{lemma}
\begin{proof}
    Taking the derivative of $\Phi(ax)-\Phi(x)$ with respect to $x$, we have
    \[\frac{\partial}{\partial x} \big[\Phi(ax)-\Phi(x)\big] = a\phi(ax) - \phi(x) = a \frac{1}{\sqrt{2\pi}}e^{-\frac{a^2x^2}{2}} - \frac{1}{\sqrt{2\pi}}e^{-\frac{x^2}{2}} = \frac{1}{\sqrt{2\pi}}e^{-\frac{x^2}{2}}\Big[ae^{-\frac{a^2-1}{2}x^2} - 1\Big]\]
    When $x > \sqrt{\frac{2\log a}{a^2-1}}$, we have, since $a > 1$,
    \[\frac{1}{\sqrt{2\pi}}e^{-\frac{x^2}{2}}\Big[ae^{-\frac{a^2-1}{2}x^2} - 1\Big] < \frac{1}{\sqrt{2\pi}}e^{-\frac{x^2}{2}}\Big[ae^{-\frac{a^2-1}{2} \big(\frac{2\log a}{a^2-1}\big)} - 1\Big] =  \frac{1}{\sqrt{2\pi}}e^{-\frac{x^2}{2}}\Big[a\cdot\frac{1}{a} - 1\Big] = 0\]
    Therefore, at $x > \sqrt{\frac{2\log a}{a^2-1}}$, $\Phi(ax)-\Phi(x)$ has strictly negative derivatives which implies it is strictly decreasing.
\end{proof}

Now we may derive the sufficient condition for $\xi_{\text{\textbf{GATE}}} > \xi_{\text{\textbf{ATE}}}$,
\begin{align*}
    &\xi_{\text{\textbf{GATE}}} > \xi_{\text{\textbf{ATE}}}&\\
    \Leftrightarrow\quad&\Phi(\delta^* + z_{\alpha/2})-\Phi(\delta^* - z_{\alpha/2}) > \Big(1-\frac{\alpha}{2}\Big) \bigg[\Phi\big(\sqrt{2}\delta^* + z_{\alpha/4}\big)-\Phi\big(\sqrt{2}\delta^* - z_{\alpha/4}\big)\bigg]&\\
    \Leftarrow\quad &\Phi(\delta^* + z_{\alpha/2})-\Phi(\delta^* - z_{\alpha/2}) >  \Phi\big(\sqrt{2}\delta^* + z_{\alpha/4}\big)-\Phi\big(\sqrt{2}\delta^* - z_{\alpha/4}\big)&\\
    \Leftrightarrow\quad &\Phi(\sqrt{2}\delta^* - z_{\alpha/4})-\Phi(\delta^* - z_{\alpha/2}) >  \Phi\big(\sqrt{2}\delta^* + z_{\alpha/4}\big)-\Phi\big(\delta^* + z_{\alpha/2}\big)&\\
    \Leftarrow\quad &\Phi(\sqrt{2}\delta^* - \sqrt{2}z_{\alpha/2})-\Phi(\delta^* - z_{\alpha/2}) >  \Phi\big(\sqrt{2}\delta^* + \sqrt{2}z_{\alpha/2}\big)-\Phi\big(\delta^* + z_{\alpha/2}\big)&(\Cref{lem:z-ratio})\\
    \Leftrightarrow\quad &\Phi[\sqrt{2}(\delta^* - z_{\alpha/2})]-\Phi[\delta^* - z_{\alpha/2}] >  \Phi\big[\sqrt{2}(\delta^* + z_{\alpha/2})\big]-\Phi\big[\delta^* + z_{\alpha/2}\big]&
\end{align*}
Since $\delta^* - z_{\alpha/2} < \delta^* + z_{\alpha/2}$, from \Cref{lem:psi-ineq}, the last inequality holds as long as $\delta^* - z_{\alpha/2} > \sqrt{\frac{2 \log(\sqrt{2})}{(\sqrt{2})^2-1}} = \sqrt{\log 2}$.  That is, a sufficient condition for $\xi_{\text{\textbf{GATE}}} > \xi_{\text{\textbf{ATE}}}$ is
\[\delta^* > \sqrt{\log 2} + z_{\alpha/2}\]
or, equivalently,
\begin{equation}
    |\delta| > \frac{2\boldsymbol{\sigma}}{\sqrt{N}}\big(\sqrt{\log 2} + z_{\alpha/2}\big)
\end{equation}

Intuitively, we see from the above condition that as the magnitude of the bias in subgroup 1 increases or the sample size $N$ increases, \textbf{GATE} will eventually have greater power than \textbf{ATE}. 
\newline
\newline
\textbf{Scenario 2}: The GATE estimators for both subgroups are biased by the same magnitude but opposite direction
\newline

This scenario can be depicted by letting $\delta_1 = \delta$ and $\delta_2 = -\delta$, $\delta \ne 0$, so that we have $\frac{\delta_1+\delta_2}{2} = 0$.  Under which the power of \textbf{ATE} and \textbf{GATE} can be given by, based on Propositions~\ref{prop:ate-power} and~\ref{prop:gate-power}: 
\begin{align}
    \xi_{\text{\textbf{ATE}}} &= 1-\left[\Phi(z_{\alpha/2}) - \Phi(-z_{\alpha/2})\right] = \alpha\\
    \xi_{\text{\textbf{GATE}}} &= 1 - \left[\Phi\Bigg(\frac{1}{\sqrt{2}}\frac{|\delta|}{\boldsymbol{\sigma}/\sqrt{N}}+z_{\alpha/4}\Bigg) - \Phi\Bigg(\frac{1}{\sqrt{2}}\frac{|\delta|}{\boldsymbol{\sigma}/\sqrt{N}}-z_{\alpha/4}\Bigg)\right]^2
\end{align}

We may give a lower bound for $\xi_{\text{\textbf{GATE}}}$:
\begin{equation}
    \xi_{\text{\textbf{GATE}}} = 1 - \left[\Phi\Bigg(\frac{1}{\sqrt{2}}\frac{|\delta|}{\boldsymbol{\sigma}/\sqrt{N}}+z_{\alpha/4}\Bigg) - \Phi\Bigg(\frac{1}{\sqrt{2}}\frac{|\delta|}{\boldsymbol{\sigma}/\sqrt{N}}-z_{\alpha/4}\Bigg)\right]^2 > 1-\left[1 - \Phi\Bigg(\frac{1}{\sqrt{2}}\frac{|\delta|}{\boldsymbol{\sigma}/\sqrt{N}}-z_{\alpha/4}\Bigg)\right]^2 
\end{equation}

Therefore, a sufficient condition for $\xi_{\text{\textbf{GATE}}}> \xi_{\text{\textbf{ATE}}}$, i.e. the power of \textbf{GATE} to be greater than \textbf{ATE} is,
\begin{equation}
    |\delta| > \frac{\boldsymbol{\sigma}}{\sqrt{N/2}}(z_{\alpha/4}+\Phi^{-1}(1-\sqrt{1-\alpha}))
\end{equation}
which can be attained with a large enough bias magnitude $|\delta|$ or large enough sample size $N$ that overcomes the penalty of multiple testing.
\newline
\newline
\textbf{Scenario 3}: The GATE estimators for both subgroups are biased by the same magnitude and direction

This scenario can be depicted by letting $\delta_1 = \delta_2 = \delta \ne 0$, so that we have $\frac{\delta_1+\delta_2}{2} = \delta$.  Under which the power of \textbf{ATE} and \textbf{GATE} can be given by, based on Propositions~\ref{prop:ate-power} and~\ref{prop:gate-power}: 
\begin{align}
    \xi_{\text{\textbf{ATE}}} &= 1-\left[\Phi\Bigg(\frac{|\delta|}{\boldsymbol{\sigma}/\sqrt{N}}+z_{\alpha/2}\Bigg) - \Phi\Bigg(\frac{|\delta|}{\boldsymbol{\sigma}/\sqrt{N}}-z_{\alpha/2}\Bigg)\right]\\
    \xi_{\text{\textbf{GATE}}} &= 1 - \left[\Phi\Bigg(\frac{1}{\sqrt{2}}\frac{|\delta|}{\boldsymbol{\sigma}/\sqrt{N}}+z_{\alpha/4}\Bigg) - \Phi\Bigg(\frac{1}{\sqrt{2}}\frac{|\delta|}{\boldsymbol{\sigma}/\sqrt{N}}-z_{\alpha/4}\Bigg)\right]^2
\end{align}

Denoting $\delta^* := \frac{|\delta|}{\boldsymbol{\sigma}/\sqrt{N}} \ge 0$, we may simplify the expressions as,
\begin{align}
    \xi_{\text{\textbf{ATE}}} &= 1-\left[\Phi\Big(\delta^*+z_{\alpha/2}\Big) - \Phi\Big(\delta^*-z_{\alpha/2}\Big)\right]\\
    \xi_{\text{\textbf{GATE}}} &= 1 - \left[\Phi\Big(\frac{1}{\sqrt{2}}\delta^*+z_{\alpha/4}\Big) - \Phi\Big(\frac{1}{\sqrt{2}}\delta^*-z_{\alpha/4}\Big)\right]^2
\end{align}

Therefore, the condition for $\xi_{\text{\textbf{ATE}}} > \xi_{\text{\textbf{GATE}}}$ is equivalent to
\begin{equation}
    g(\delta^*) := \Big[\Phi\Big(\frac{1}{\sqrt{2}}\delta^* + z_{\alpha/4}\Big)-\Phi\Big(\frac{1}{\sqrt{2}}\delta^* - z_{\alpha/4}\Big)\Big]^2 - \Big[\Phi\Big(\delta^* + z_{\alpha/2}\Big) - \Phi\Big(\delta^* - z_{\alpha/2}\Big)\Big]>0
\end{equation}

A graph for $g(\delta^*)$ with $\alpha = 0.005, 0.01, 0.05, 0.1$ is shown in \Cref{fig:gfunction}, which demonstrates that $g(\delta^*) > 0$ is satisfied for any $\delta^* > 0$.  Therefore, under the scenario where a common bias is shared across subgroups, the power of \textbf{ATE} is greater than \textbf{GATE} irrespective of the magnitude of bias. 

\begin{figure}[ht]
    \centering
    \includegraphics[width = 0.65\textwidth,trim={1cm 1cm 1cm 1cm}, clip]{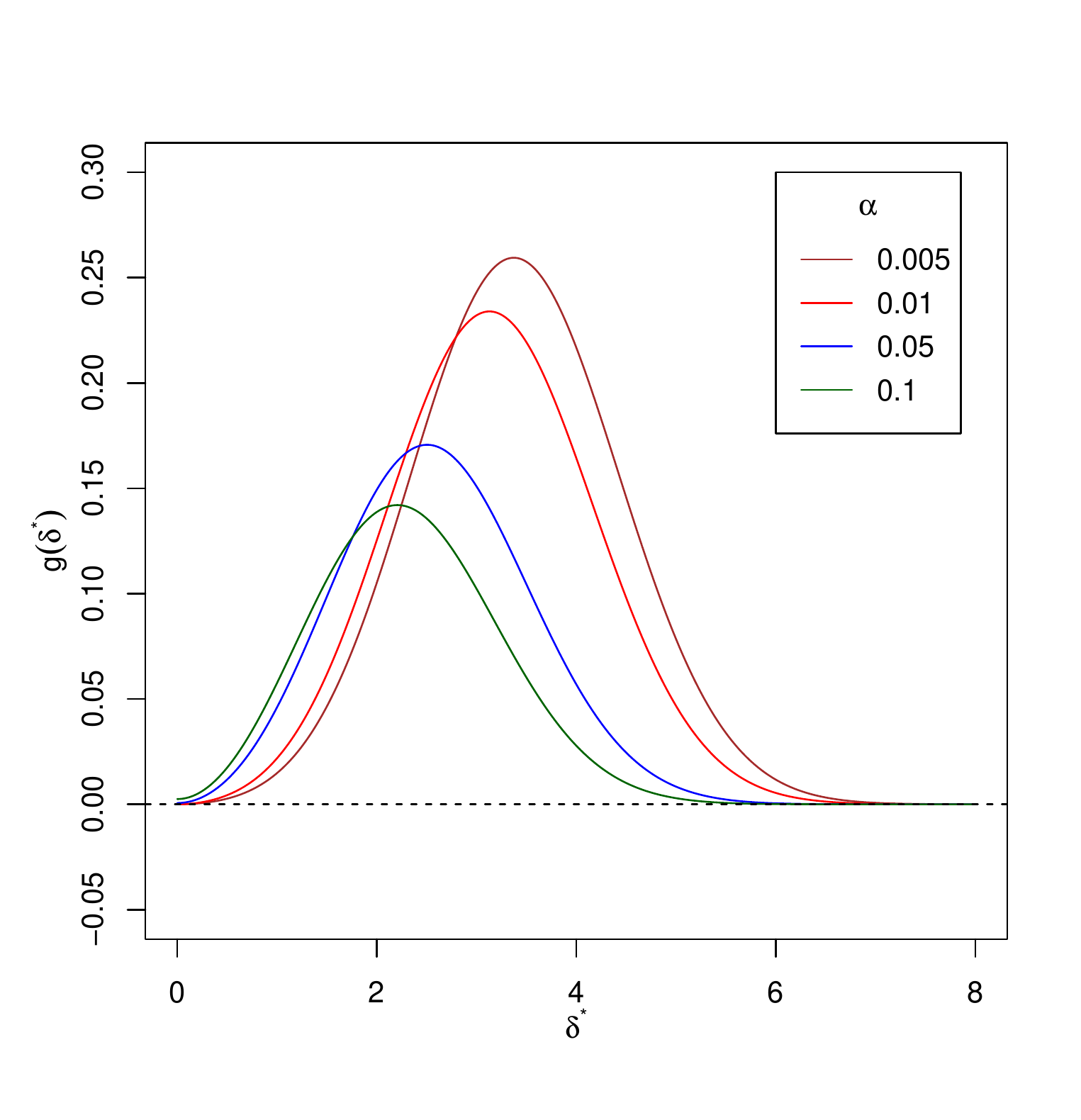}
    \caption{Plot of function $g(.)$ under $\alpha = 0.005, 0.01, 0.05$ and $0.1$}
    \label{fig:gfunction}
\end{figure}
\medskip
Overall, we find that the relative asymptotic power of \textbf{ATE} and \textbf{GATE} depends on the homogeneity of bias amongst the subgroups and the magnitude of the bias, and should be analyzed on a case-by-case basis.

\end{document}

%% file: paperSpecificNotation.tex
\newcommand{\tauki}{\tau_i{(k)}}

\newcommand{\tauhki}{\hat{\tau}_i{(k)}}